\documentclass{article}
\usepackage{color}
\usepackage{graphicx, amsthm, amsmath, amssymb} 
\usepackage{tikz}
\usepackage[font={sf}]{caption}
\definecolor{purple}{rgb}{0.5,0.0,0.5}
\usepackage[backref=page]{hyperref}
\hypersetup{
    unicode=false,          
    colorlinks=true,        
    linkcolor=blue,         
    citecolor=purple,       
    filecolor=magenta,      
    urlcolor=cyan           
}

\newcommand{\review}[1]{\textcolor{black}{#1}}

\usepackage{algorithm}
\usepackage[noend]{algpseudocode}
\usepackage{enumerate}
\usepackage{thmtools,thm-restate}

\usepackage[capitalise,noabbrev, nameinlink]{cleveref}

\newtheorem{theorem}{Theorem}[section]
\newtheorem{lemma}[theorem]{Lemma}
\newtheorem{definition}[theorem]{Definition}
\newtheorem{observation}[theorem]{Observation}
\newtheorem{question}[theorem]{Question}
\newtheorem{proposition}[theorem]{Proposition}

\newtheorem{conjecture}[theorem]{Conjecture}

\newcommand{\eps}{\varepsilon}
\newcommand{\poly}{\mathrm{poly}}
\newcommand{\N}{{\mathbb{N}}}
\newcommand{\R}{{\mathbb{R}}}
\newcommand{\J}{\mathcal{J}}

\newcommand{\tO}{\tilde{O}}
\newcommand{\tS}{\tilde{S}}
\newcommand{\cC}{\mathcal{C}}
\newcommand{\cR}{\mathcal{R}}

\newcommand{\rb}[1]{\left( #1 \right)}
\newcommand{\bb}[1]{\left[ #1 \right]}

\usepackage{todonotes}
\usepackage{color}
\usepackage{hyperref}

\definecolor{greenblue}{rgb}{0.7,0.9,1.0}
\definecolor{lightgreen}{rgb}{0.8,1,0.8}

\usepackage[margin=1in]{geometry}
\title{Approximate Counting of Permutation Patterns}

\author{Omri Ben-Eliezer\thanks{Faculty of Computer Science, Technion -- Israel Institute of Technology, Haifa, Israel.  
Supported by an Alon Scholarship and by a Taub Family Foundation “Leaders in Science \& Technology” Fellowship. Work conducted in part while the author was at MIT and later at the Simons Institute for the Theory of Computing. Email: \texttt{omribene@cs.technion.ac.il}}
\and
Slobodan Mitrovi{\'c}\thanks{Department of Computer Science, University of California, Davis, CA, USA. Supported by the Google Research Scholar and NSF Faculty Early Career Development Program No.~2340048. Part of this work was conducted while the author was visiting the Simons Institute for the Theory of Computing. Email: \texttt{smitrovic@ucdavis.edu}}
\and
Pranjal Srivastava\thanks{Massachusetts Institute of Technology, Cambridge, MA, USA. Email: \texttt{pranjal0@mit.edu}}
}
\date{}

\begin{document}

\maketitle

\begin{abstract}
We consider the problem of counting the copies of a length-$k$ pattern $\sigma$ in a sequence $f \colon [n] \to \R$, where a copy is a subset of indices $i_1 < \ldots < i_k \in [n]$ such that $f(i_j) < f(i_\ell)$ if and only if $\sigma(j) < \sigma(\ell)$. 
This problem is motivated by a range of connections and applications in ranking, nonparametric statistics, combinatorics, and fine-grained complexity, especially when $k$ is a small fixed constant. 

Recent advances have significantly improved our understanding of counting and detecting patterns.
Guillemot and Marx [2014] obtained an $O(n)$ time algorithm for the detection variant for any fixed $k$. 
Their proof has laid the foundations for the discovery of the twin-width, a concept that has notably advanced parameterized complexity in recent years. 
Counting, in contrast, is harder: 
it has a conditional lower bound of $n^{\Omega(k / \log k)}$ [Berendsohn, Kozma, and Marx, 2019] and is expected to be polynomially harder than detection as early as $k = 4$, given its equivalence to counting $4$-cycles in graphs [Dudek and Gawrychowski, 2020].

In this work, we design a deterministic near-linear time $(1+\varepsilon)$-approximation algorithm for counting $\sigma$-copies in $f$ for all $k \leq 5$. 
Combined with the conditional lower bound for $k=4$, this establishes the first known separation between approximate and exact pattern counting. 
\review{Interestingly, while neither the sequence $f$ nor the pattern $\sigma$ are monotone, our algorithm makes extensive use of coresets for monotone functions [Har-Peled, 2006].}
Along the way, we develop a near-optimal data structure for $(1+\varepsilon)$-approximate increasing pair range queries in the plane, which exhibits a conditional separation from the exact case and may be of independent interest. 
\end{abstract}
\setcounter{page}{0}
\thispagestyle{empty}

\newpage

\tableofcontents
\setcounter{page}{0}
\thispagestyle{empty}

\newpage

\section{Introduction}
Detecting and counting structural patterns in a data sequence is a common algorithmic challenge in various theoretical and applied domains. Some of the numerous application domains include ranking and recommendation \cite{DKNS2001}, time series analysis \cite{PhysRevLett.88.174102}, and computational biology \cite{CombinatoricsGenomeRearrangements}, among many others. On the mathematical/theoretical side, problems involving sequential pattern analysis naturally arise, e.g., in algebraic geometry \cite{AbeBilley2016}, combinatorics \cite{crudele2023permutation,Grubel2023}, and nonparametric statistics \cite{EvenZoharLeng2021}. 

Formally, we are interested here in finding \emph{order patterns}, \review{which are also called \emph{permutation patterns}, and} defined as follows. Given a real-valued sequence $f \colon [n] \to \R$ and a permutation pattern $\sigma \colon [k] \to [k]$, a \emph{copy} of the pattern $\sigma$ in the sequence $f$ is any subset of $k$ indices $i_1 < i_2 < \ldots < i_k$ so that for $j, \ell \in [k]$, $f(i_j) < f(i_\ell)$ if and only if $\sigma(j) < \sigma(\ell)$; 
see \Cref{fig:order_patterns} for a visual depiction.

\begin{figure}[h]
\begin{center}
\begin{tikzpicture}[scale=0.5]
    \draw[thin,->] (0,0) -- (13,0) node[right] {$x$}; 
    \draw[thin,->] (0,0) -- (0,5) node[left] {$f(x)$}; 

    \node[draw, circle, inner sep=1.5pt, fill=white] at (0.5,1) {};
    \node[draw, circle, inner sep=1.5pt, fill=white] at (1,2) {};
    \node[draw, circle, inner sep=1.5pt, fill=white] at (1.5,1.3) {};
    \node[draw, circle, inner sep=1.5pt, fill=white] at (2,4) {};
    \node[draw, circle, inner sep=1.5pt, fill=white] at (2.5,2.8) {};
    \node[draw, circle, inner sep=1.5pt, fill=white] at (3,1.9) {};
    \node[draw=black, circle, inner sep=1.5pt, fill=black] at (3.5,1.15) {};
    \node[draw, circle, inner sep=1.5pt, fill=white] at (4,3.5) {};
    \node[draw, circle, inner sep=1.5pt, fill=white] at (4.5,2.9) {};
    \node[draw, circle, inner sep=1.5pt, fill=white] at (5,1.35) {};
    \node[draw, circle, inner sep=1.5pt, fill=white] at (5.5,0.4) {};
    \node[draw, circle, inner sep=1.5pt, fill=white] at (6,1.2) {};
    \node[draw, circle, inner sep=1.5pt, fill=white] at (6.5,2.1) {};
    \node[draw, circle, inner sep=1.5pt, fill=white] at (7,1.5) {};
    \node[draw=black, circle, inner sep=1.5pt, fill=black] at (7.5,4.2) {};
    \node[draw, circle, inner sep=1.5pt, fill=white] at (8,2.2) {};
    \node[draw=black, circle, inner sep=1.5pt, fill=black] at (8.5,3) {};
    \node[draw, circle, inner sep=1.5pt, fill=white] at (9,0.8) {};
    \node[draw, circle, inner sep=1.5pt, fill=white] at (9.5,3.5) {};
    \node[draw, circle, inner sep=1.5pt, fill=white] at (10,3.2) {};
    \node[draw, circle, inner sep=1.5pt, fill=white] at (10.5,1.6) {};
    \node[draw, circle, inner sep=1.5pt, fill=white] at (11,0.5) {};
    \node[draw=black, circle, inner sep=1.5pt, fill=black] at (11.5,2) {};
    \node[draw, circle, inner sep=1.5pt, fill=white] at (12,0.9) {};
    \node[draw, circle, inner sep=1.5pt, fill=white] at (12.5,1.5) {};
\end{tikzpicture}
\end{center}
\vspace{-0.2cm}
\caption{A configuration of $n$ points in two dimensions (with no two points sharing the same $x$ coordinate), represented as a function $f \colon [n] \to \R$. The four full points form a copy of the permutation pattern $1432$.}
\label{fig:order_patterns}
\end{figure}
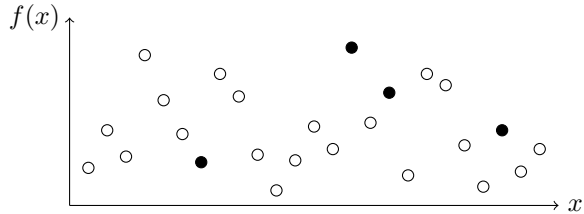

In the permutation pattern matching (PPM) problem,\footnote{We shall interchangeably use the terms ``pattern matching'' and ``pattern detection'' to refer to this problem.} the task is to determine whether $f$ contains at least one copy of the pattern $\sigma$. 
In the counting variant, the goal is to return the exact or approximate number of $\sigma$-copies in $f$. Recent years have seen several breakthroughs in both detection and counting, revealing important implications in parameterized and fine-grained complexity.

Of most importance is the case where $k$ is a small constant, which has a large number of diverse applications and interesting connections:
\begin{itemize}
    \item Counting inversions, which are $21$-copies ($k=2$) is of fundamental importance for ranking applications \cite{DKNS2001}. It has thus attracted significant attention from the algorithmic community for the last several decades, for both exact counting \cite{ChanPatrascu2010,Dietz1989,FredmanSaks1989} and approximate counting \cite{ChanPatrascu2010, AnderssonPetersson98}. 
    \item Counting $4$-patterns\footnote{We henceforth use the abbreviation ``$k$-pattern'' to refer to a permutation pattern of length $k$.} is equivalent, by a bidirectional reduction, to counting $4$-cycles in sparse graphs. The latter is a fundamental problem in algorithmic graph theory (e.g., \cite{Alon1997Finding, DKS17}) and fine-grained complexity (e.g., \cite{WWWY15,ABKZ22,ABF23,JX23}). 
    This equivalence was shown by Dudek and Gawrychowski \cite{DudekGawrychowski2020}. 
    \item Pattern counting for fixed $k$ (especially $k \leq 5$) has deep and intricate connections to (bivariate) \emph{independence testing},
    a fundamental question in nonparametric statistics that asks the following. Given $n$ pairs of samples $(x_1, y_1), \ldots, (x_n, y_n)$ from two real continuous random variables $X$ and $Y$, should we deduce that $X$ and $Y$ are independent? 
    
    This question has seen a long line of work in non-parametric statistics, e.g., \cite{evenzohar2020independence, BergsmaDassios2014, Yanagimoto1970, Chatarjee2021, BlumKieferRosenblatt1961}. 
    A line of work that started by Hoeffding in the 1940's \cite{Hoeffding1948} and is still very active to this day establishes distribution-free methods to test independence by (i) ordering the sample pairs according to the values of the $x_i$'s, effectively treating the $y_i$'s as a length-$n$ sequence; and (ii) deciding whether $X$ and $Y$ are independent based only on the \review{number of occurrences of all $k$-patterns for $k \leq 5$} in this sequence. This is a special case of the much broader notion of $U$-statistics \cite{Lee1990UStatisticsTA,koroljuk1994theory}. 
    See \cite{evenzohar2020independence,Grubel2023} for more details on this fascinating connection.
        
    \item A family of length-$n$ permutations is considered quasirandom if, roughly speaking, the number of occurrences of every pattern in the family (of any length) is asymptotically similar to that of a random permutation. Quasirandomness turns out to be quite closely related to independence testing, discussed above, and it is known that the counts of patterns of length up to four suffice to determine quasirandomness, see, e.g., \cite{crudele2023permutation,Grubel2023}.
    \item Permutation pattern matching allows one to deduce whether an input $f$ is free from some pattern $\sigma$, and consequently run much faster algorithms tailored to $\sigma$-free instances.
    For example, classical optimization tasks such as binary search trees and $k$-server become much faster on $\sigma$-free inputs \cite{berendsohn2023optimization}, and a recent fascinating result by Opler~\cite{opler2024sorting} shows that sorting can be done in \emph{linear} time in pattern-avoiding sequences.\footnote{\review{Notably, Opler's result works without even knowing the pattern, so it is not needed to run PPM as a preliminary step here.}}
    Pattern matching itself sometimes also becomes faster in classes of $\sigma$-free permutations \cite{JelinekKyncl2017,JOP2021,BoseBussLubiw98}. 
    \end{itemize}

Consequently, there has been a long line of computational work on pattern matching and counting, e.g., \cite{BrunerLackner2012,BergsmaDassios2014,JelinekKyncl2017,berendsohn2021finding,evenzohar2020independence,JOP2021,Chatarjee2021,GawrychowskiRzepecki2022}. Here, we focus on the most relevant results in the constant $k$ case. Notably, the version of the problem where $k$ is large (linear in $n$) is NP-hard \cite{BoseBussLubiw98}.

Both matching and counting admit a trivial algorithm with running time $O(k n^{k})$: the idea is to enumerate over all $k$-tuples of indices in $f$, and check if each such tuple in $f$ induces a copy of the pattern. But can these algorithmic tasks be solved in time substantially smaller than $n^k$? 

\paragraph{Pattern matching: a linear-time algorithm, and the twin-width connection.}
In the matching case, the answer is resoundingly positive. The seminal work of Guillemot and Marx \cite{GuillemotMarx2014} shows that PPM is a fixed parameter tractable (FPT) problem that takes $O(n)$ time for fixed $k$.\footnote{Unless mentioned otherwise, the computational model is Word RAM, that allows querying a single function value or comparing two values in constant time.} Their running time is of the form $2^{O(k^2 \log k)} \cdot n$; the bound was slightly improved by Fox to $2^{O(k^2)} \cdot n$ \cite{fox2013stanleywilf}.

The technical argument of \cite{GuillemotMarx2014} relies on two main ingredients: the first is the celebrated result of Marcus and Tardos \cite{MarcusTardos2004} in their proof of the Stanley-Wilf conjecture \cite{FurediHajnal1992, Klazar2000}, while the second is a novel width notion for permutations suggested in their work. This newly discovered width notion subsequently led to the development of the very wide and useful notion of \emph{twin-width}, which has revolutionized parameterized complexity in recent years. Indeed, the work of Bonnet, Kim, Thomassé, and Watrigant \cite{BonnetTwinWidth2021}, which originally defined twin-width, begins with the following statement: \textit{``Inspired by a width invariant defined on permutations by Guillemot and Marx \cite{GuillemotMarx2014}, we introduce the notion of twin-width on graphs and on matrices.''}

\paragraph{Pattern counting: algorithms and hardness.}
Exact counting, meanwhile, is unlikely to admit very efficient algorithms. A series of works from the last two decades has gradually improved the $n^{k}$ upper bound, obtaining bounds of the form $n^{(c+o(1))k}$ for constant $c < 1$ \cite{AAAH2001, AhalRabinovich2008}. The current state of the art, proved by Bernedsohn, Kozma, and Marx \cite{berendsohn2021finding} is of the form $n^{k/4 +o(k)}$. The same work shows, however, that $n^{o(k/\log k)}$-time algorithms for exact counting cannot exist unless the exponential-time hypothesis (ETH) is false. In fact, a similar impossibility result (with a slightly weaker bound of $n^{o(k / \log^2 k)}$) holds, under ETH, even for exact pattern counting where the pattern $\sigma$ itself avoids some fixed (larger) pattern \cite{jelinek_IPEC}.
The above results treat $k$ as a variable; we next focus on the case where $k$ is very small, given the myriad of applications discussed before.

In the case $k=2$, it is easy to obtain an exact counting algorithm in time $O(n \log n)$ (in the Word RAM model), via a variant of merge sort. A line of work \cite{Dietz1989,FredmanSaks1989,AnderssonPetersson98, ChanPatrascu2010} sought to obtain improved algorithms for both exact and approximate counting (to within a $1+\epsilon$ multiplicative factor).\footnote{Formally, a $(1+\epsilon)$-approximate counting algorithm is required, given access to a pattern $\sigma$ and a function $f$, to return a value between $X / (1+\eps)$ and $(1+\eps)X$, where $X$ is the number of $\sigma$-copies in $f$.} The best known exact and approximate upper bounds for $k=2$ are $O(n \sqrt{\log n})$ and $O(n)$, respectively, both proved by Chan and P\u{a}tra\c{s}cu \cite{ChanPatrascu2010}.

The cases of $k=3$ and $k=4$ have been the subject of multiple recent works.
Even-Zohar and Leng \cite{EvenZoharLeng2021} developed an object called \emph{corner tree} to count a family of patterns (that slightly differ from permutation patterns) in near-linear time. Using linear combinations of corner tree formulas, they obtained near-linear time algorithm for all patterns of length $3$ and some (8 out of 24) length-$4$ patterns. For the remaining ones of length $4$, the same work obtains an $O(n^{3/2})$ time algorithm using different techniques. This interesting dichotomy between ``easy'' and ``hard'' $4$-patterns raises an interesting question: \emph{is the dichotomy an artifact of the specific technique, or is there an inherent computational barrier?}

Dudek and Gawrychowski \cite{DudekGawrychowski2020} proved that the latter is true: exact counting of any ``hard'' $4$-pattern is equivalent (via bidirectional reductions) to exact counting of $4$-cycles in graphs, a central and very well studied problem in algorithmic graph theory. The concrete equivalence stated in their paper (see Theorem 1 there) is that an $\tilde{O}(m^\gamma)$-time algorithm for counting $4$-cycles in $m$-edge graphs implies an $\tilde{O}(n^{\gamma})$ time algorithm for counting ``hard'' 4-patterns, and vice versa. While this has led to a slightly improved $O(n^{1.48})$ upper bound based on best known results for counting $4$-cycles in sparse graphs \cite{WWWY15}, the more interesting direction to us is the lower bound side.
A line of recent works obtains conditional lower bounds on $4$-cycle counting, that apply already for the easier task of $4$-cycle detection \cite{ABKZ22,ABF23,JX23}. These works imply that conditioning on the Strong 3-SUM conjecture, detecting whether a (sufficiently sparse) graph with $m$ edges contains a $4$-cycle requires $m^{1+\Omega(1)}$ time (see, e.g., the discussion after Theorem 1.14 in \cite{JX23}), which translates to an $n^{1+\Omega(1)}$ lower bound for exact counting $4$-patterns, via \cite{DudekGawrychowski2020}. 

\subsection{Our Results}
Given the separation between the $O(n)$ complexity of pattern detection and the $n^{1+\Omega(1)}$ conditional lower bound for pattern counting already for $k=4$, and the importance of counting in the constant-$k$ regime, we ask whether approximate counting can be performed in time substantially (polynomially) faster than exact counting.
\begin{center}
\emph{What is the computational landscape of $(1+\eps)$-approximate counting of $k$-patterns, for small fixed $k$, as compared to exact counting and matching? Is approximate counting much faster than exact counting?}
\end{center}

The only case where the best known $(1+\eps)$-approximate algorithm is faster than the best known exact algorithm is when $k=2$ \cite{ChanPatrascu2010}, but the gap is only of order $\sqrt{\log n}$ (i.e., between $O(n)$ and $O(n \sqrt{\log n})$), and no nontrivial exact counting lower bounds are known.
Thus, it remains unknown whether exact counting is harder than approximate counting even for $k=2$, and even if it is, the gap would be of lower order.

Our main contribution, stated below, is a near-linear time approximate counting algorithm for $k \leq 5$.

\begin{restatable}{theorem}{maintheorem}
\label{thm:main_theorem}
For every permutation pattern $\sigma$ of length $k \leq 5$ and every $\eps > 0$, the following holds.
There exists a deterministic algorithm that, given access to a function $f \colon [n] \to \R$, returns the number of $\sigma$-copies in $f$, up to a multiplicative error of $1+\eps$, in time $n \cdot \left(\eps^{-1}\log n\right)^{O(1)}$.
\end{restatable}

Combined with the $n^{1+{\Omega(1)}}$ lower bound for counting ``hard'' $4$-patterns (e.g., $2413$), our result implies a \emph{polynomial} separation between exact and $(1+\eps)$-approximate algorithms for 4-patterns and 5-patterns. 

We stress that our algorithm is \emph{deterministic}, and that prior to our work no near-linear algorithm -- deterministic or randomized -- was known to exist for this problem for $k=4$ and $k=5$. Our main new technique is a suitable Birg\'e-type lemma, which is deterministic and can be thought as a Riemann integration argument on monotone functions. We discuss our techniques in \Cref{sec:techniques}, and other approaches in the approximate counting literature, as well as some of the many open questions that arise following this work, in \Cref{sec:discuss}.
Notably, existing approaches from the literature are (i) inherently randomized and (ii) do not seem to apply directly to our setting; for them to work, one has to use introduce new components, which are variants of the techniques that we develop here. For further discussion, see~\Cref{sec:discuss}.

The paper includes a full, self-contained proof for $k = 4$; recall that for $k \leq 3$, even exact counting algorithms have near-linear time complexity \cite{EvenZoharLeng2021}. 
For $k=5$, our proof is computer-assisted: the algorithm enumerates over multiple parameter choices and techniques, heavily depending on the pattern structure.
Verifying that the algorithm works for all patterns requires a tedious case analysis for $k=5$, involving 512 cases, each of which is straightforward to verify based on the output from our code.
We describe the set of techniques used, establish how they can be combined, provide examples of typical use cases, and delegate the full enumeration to the software.
The source code for the enumeration and the full output (including for $k=5$) are provided here:
\textsf{\href{https://github.com/omribene/approx-counting}{https://github.com/omribene/approx-counting}}. 


Our proof can be immediately adapted to provide an algorithm for \emph{enumerating} (or \emph{listing}) copies of the pattern. In the enumeration problem, we are given $f, \sigma$, and an integer $t$, and are required to provide a list of $t$ copies of $\sigma$ in $f$ (or the full list if there are less than $t$ copies). We obtain the following result.
\begin{theorem}\label{thm:enumeration}
For every permutation pattern $\sigma$ of length $k \leq 5$ and every $t \in \N$, the following holds.
There exists a deterministic algorithm that, given access to a function $f \colon [n] \to \R$, returns a list of $t$ copies of $\sigma$ in $f$ (or all such copies, if there are fewer than $t$), in time $(n+t) \cdot \log^{O(1)}n$.
\end{theorem}
Our results further highlight the contrasting behavior between $4$-cycles in sparse graphs and $4$-patterns in sequences. The exact counting complexities for these objects are equal, due to the linear-size bidirectional reductions between these problems \cite{DudekGawrychowski2020}. 
Meanwhile, for detection we have a separation between the $O(n)$ algorithm for patterns \cite{GuillemotMarx2014} and the $n^{1+\Omega(1)}$ conditional lower bound for cycles in sparse graphs \cite{ABKZ22,ABF23,JX23}.
Since the same lower bound also applies to approximate counting of $4$-cycles, this implies a separation for approximate counting. 
Finally, for enumeration the lower bounds of \cite{ABF23,JX23} are stronger (and in fact tight), of order $\Omega(\min\{n^{2-o(1)}, m^{4/3-o(1)}\})$, conditioning on the 3-SUM conjecture.
Again, since enumeration of permutation patterns takes time near-linear in $n$ and $t$, we get a separation from $4$-cycles for sufficiently small values of $t$.

Our proof of the main approximate pattern counting result builds on a number of simple and elementary tools, most of which will be discussed in~\Cref{sec:techniques}. Among these is a seemingly new deterministic data structure for approximately counting $12$-copies in arbitrary axis-aligned rectangle ranges in the plane, which we believe may be of independent interest. 

\begin{proposition}[Data structure for approximate $12$-counting queries]
\label{lemma:counting-12-copies}
    There exists a deterministic data structure for $n$-point sets in $\R^2$, that implements the following with preprocessing time $\tO(n \eps^{-1})$. Given an axis-parallel rectangle $R \subseteq \R^2$ as a query, the data structure reports a $(1+\eps)$-approximation of the number of $12$-copies inside $R$, with per-query time $\poly(\log n) \cdot \eps^{-1}$.
\end{proposition}
\ref{lemma:counting-12-copies} gives rise to an interesting conditional separation between exact and approximate data structures for $12$-counting in axis aligned rectangle ranges. Duray, Kleiner, Polak, and Williams \cite{DKPW20} proved bidirectional reductions showing the following two problems are equivalent (up to lower order terms): 
\begin{itemize}
\item Given a set of $n$ points in the plane together with $q = \Theta(n)$ vertical ``slabs'' \review{(which are possibly overlapping)}, compute the exact number of $12$-copies in each of the slabs; and
\item ``Edge-Triangle Counting'': Given access to an unknown graph on $n$ vertices, count for each edge $e$ the number of triangles containing $e$.
\end{itemize}
The first task can be trivially implemented by constructing an exact data structure for $12$-counting queries (in the analogous setting to \Cref{lemma:counting-12-copies}), and using $\Theta(n)$ queries to this data structure to compute the exact $12$-count in each slab, since such slabs are a special case of a rectangle. Thus, lower bounds for the first task translate to lower bounds for the exact analogue of the $12$-range queries data structure.

The results of \cite{DKPW20} imply that conditioned on the 3-SUM conjecture, both of these tasks cannot be solved in time better than $n^{\frac{4}{3}-o(1)}$. In particular, conditioned on 3-SUM, this rules out the existence of exact data structures (for the same $12$-copies range counting problem as in~\Cref{lemma:counting-12-copies}) that simultaneously have construction time better than $n^{\frac{4}{3}-\Omega(1)}$ and query time better than $n^{\frac{1}{3}-\Omega(1)}$ for the exact variant. Compared to our approximate data structure, with near-linear construction time and polylogarithmic query time, this gives a conditional \emph{polynomial separation} between exact and approximate data structures.

Notably, there is extensive body of literature on approximate range queries in computational geometry, e.g., \cite{AronovHarpeled2008,afshani2009approximate,AHZ09,kaplan2011range,ChanWilkinson2016,Rahul2017}. The main line of work is on data structures for counting the number of points within a range (i.e., counting ``1-copies'' in our language), as contrasted to our data structure, which counts $12$-copies, also known as increasing pairs or dominance pairs, inside the range. 

For point queries in axis-aligned rectangles (a task called ``orthogonal range counting'' in computational geometry), even exact queries can be easily supported with $\tilde{O}(n)$ construction time and polylogarithmic query time, using textbook segment trees; optimal constructions bring the query complexity down to $o(\log n)$, see, e.g., \cite{ChanWilkinson2016}. While approximate approaches here may provide further savings as compared to their exact counterparts, the separation is at worst polylogarithmic (and not polynomial). The polynomial separation between our approximate data structure from~\Cref{lemma:counting-12-copies} and the conditional lower bounds of \cite{DKPW20} for $12$-range queries is thus much more dramatic compared to the analogous situation in point range queries.

\subsection{Our Techniques}
\label{sec:techniques}
We next outline the three main ideas central to our approximate counting work: (i) the Birg\'e technique for exploiting structural monotonicity; (ii) using separators to impose additional structure on pattern instances; and (iii) a specialized data structure for approximating the counts of $12$ copies within axis-parallel rectangles.\footnote{Throughout our work, we assume the input is a permutation. Nevertheless, our proofs also handle inputs/functions that contain points with the same $y$-coordinate, i.e., the proofs tolerate  $f(i) = f(j)$ for $i \neq j$.
Also, without loss of generality, for the problem of counting patterns, it can be assumed that $f(i) \in \{0, 1, \ldots, n\}$.
}

\subsubsection{Leveraging the Birg\'e decomposition for monotonicity-based counting (\Cref{sec:almost-all-4-length})}
\review{Our proof makes crucial use of the Birg\'e approximation method. This is a simple method to approximate monotone distributions using a step function with few steps. It was developed by Lucien Birg\'e in the 1980's \cite{Birge1987} (see also \cite{DDSVV13,DDS14,CanonneSurvey} for applications in distribution testing). In our context, we use a discrete version of this result: a coreset for monotone functions, proved by Har-Peled \cite{HarPeled2006} (see also \cite{Halman2014}), which states the following. To approximate the sum of a (weakly) monotone sequence $x_1 \geq x_2 \geq \ldots \geq x_n \geq 0$ to within a $1+\eps$ multiplicative factor, one only needs to query a sublinear number, $O(\eps^{-1}\log n)$, of the elements in the sequence.}

\review{While it is not clear at first sight where one can find monotone quantities in our context, the use of coresets for monotone functions proves unexpectedly powerful and useful for our algorithms.
We illustrate this idea in counting $4$-patterns.}
Our approach to approximating the count of patterns like $1324$ starts by fixing a value of the ``3''.
Specifically, we divide the set of all $1324$-copies in the permutation based on the position of ``3'', creating subsets $C_1, C_2, \ldots, C_n$ where for each $1324$-copy in $C_i$, the ``3'' appears at the $i$-th location.

Once ``3'' is fixed to a certain position, we look at the possible positions for ``4''. Fixing ``4'' further organizes $C_i$ into smaller groups based on the placement of ``4'' relative to ``3''. One such scenario is illustrated in \Cref{fig:fix-3}.
\begin{figure}[h]
    \centering
    \includegraphics[width=0.5\linewidth]{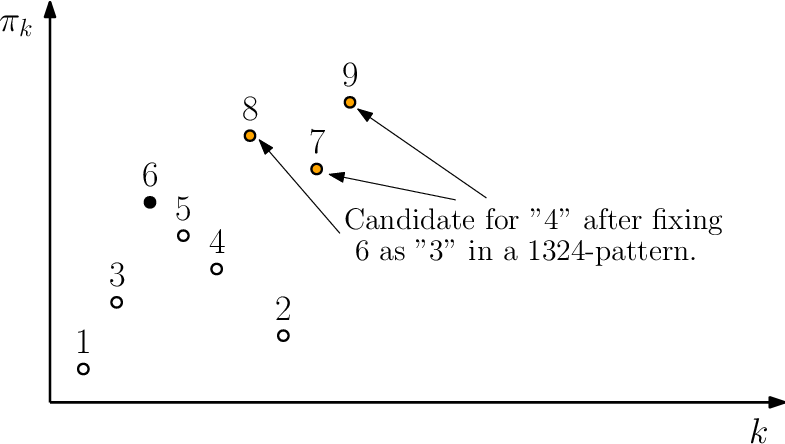}
    \caption{The illustration corresponds to permutation $\pi = 136548279$, depicted in a plane at points $(i, \pi_i)$.}
    \label{fig:fix-3}
\end{figure}
The key insight is that each position of ``4'' constrains the remaining elements of the $1324$ patterns in a \emph{monotone} way. 
For example, after fixing ``3'' to a specific position in the permutation, we can identify all positions of ``4'' that can extend this configuration into valid $1324$ copies. Within this subset, the \emph{positions} of ``4'' exhibit a specific ordering: if ``4'' appears at a given position in the sequence, any more-to-the-right occurrence of ``4'' will continue to yield valid $1324$ copies!
Similarly, we fix ``2'' and then count the relevant candidates for ``1''. In \Cref{sec:almost-all-4-length}, we show that fixing ``2'' also exhibits a certain monotonicity.

We use the Birg\'e decomposition to take advantage of this structure. 
The decomposition allows us to break down each subset $C_i$ into manageable, monotone classes and then efficiently approximate the count of each class in polylogarithmic time. 
By structuring the count around this monotonicity, we can approximately compute each $|C_i|$ without directly enumerating all possibilities, which would be computationally expensive.

So, by fixing values like ``3'', then ``4'', and then ``2'', and using the Birg\'e decomposition to handle the emerging monotonic structures, we reduce the complexity of counting $1324$ patterns to a series of fast approximations, leading to $O(n \cdot \poly(\log n, \eps^{-1}))$ running time.

\subsubsection{Imposing structure through separators for $4$-patterns (\Cref{sec:2413})}
While the Birg\'e decomposition effectively handles some patterns, others (such as $2413$) do not exhibit the same straightforward monotonic structure. 
For these patterns, we introduce separators to impose additional structural constraints.

Consider the $4$-pattern $2413$. Unlike $1324$, this pattern does not naturally exhibit a straightforward monotonic structure. 
If we fix ``4'' to a particular position, we would ideally like the positions of other elements -- ``2'', ``1'', and ``3'' -- to show some consistent ordering so that we can apply an efficient counting method. 
However, without further structuring, the placements of ``1'' and ``3'' relative to ``4'' do not seem to reveal any particular order.

\begin{figure}[h]
    \centering
    \includegraphics[width=0.8\linewidth]{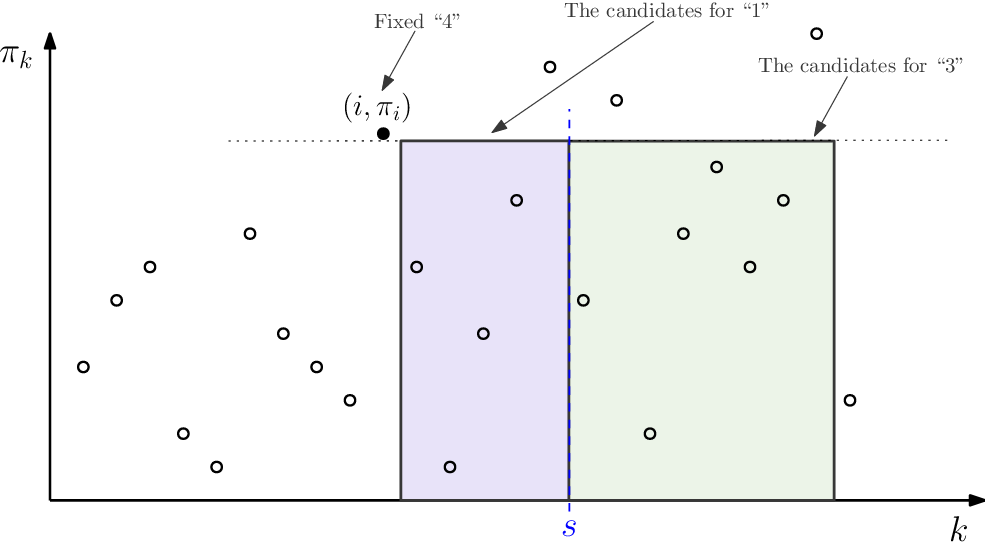}
    \caption{
        An illustration of the use of separators to split the candidates for ``1'' and ``3'' into disjoint but neighboring regions, based on their position.
    }
    \label{fig:separator-example}
\end{figure}
To handle this, we introduce a separator to divide the possible positions of elements in $2413$ based on their relative positions to ``4''. 
For instance, after fixing ``4'', we introduce a position-based separator $s$ that splits the plane into two regions.
We then require that ``1'' appears to the left of $s$ while ``3'' appears to the right of $s$. This allows us to approximate the count of $2413$ copies within each configuration independently. We illustrate such a separator in \Cref{fig:separator-example}.
With this separator in place, the counts of $2413$ copies become monotone again, enabling us to apply the Birg\'e decomposition to each subset created by the separator.
The complete analysis is presented in \Cref{sec:2413}.

\subsubsection{A data structure for $(1+\eps)$-approximate $12$-range queries in plane (\Cref{sec:counting-12-copies})}

We introduce a deterministic data structure for approximate counting simple $12$ patterns (increasing pairs) within arbitrary axis-aligned rectangles. 
This primitive allows us to query the approximate number of $12$ copies within any subregion of the input permutation. We employ this data structure to count $5$-patterns.

To develop this $12$-copy counting data structure, we use a two-dimensional segment tree.
With this tree, we pre-process the points in a bottom-up manner in $O(n \cdot \poly(\log n, \eps^{-1}))$ time. 
\Cref{sec:counting-12-copies} details the implementation of this bottom-up pre-processing.
This pre-processing computes an approximate number of $12$ copies within each vertex of the segment tree.
These pre-computed values are later used to answer queries for approximating the number of $12$ copies within arbitrary rectangles, each answered in polylogarithmic time. 

\subsubsection{Global separators for $5$-patterns (\Cref{sec:5-patterns})}
When extending our approach to $5$-patterns, we introduce an enhanced separator structure, which we refer to as global separators. 
This structure is specifically designed for handling the additional complexity that arises when counting 5-patterns, such as $24135$.

These separators are easiest to describe using the language of two-dimensional segment trees. Consider a two-dimensional segment tree $S$ built over the plane.
The outer segment tree divides the space along the $x$-axis, while each vertex in this tree contains an inner segment tree that further partitions the range along the $y$-axis. 

For each vertex $v$ in the outer segment tree, we want to count all copies of a given $5$-pattern, e.g., $24135$, that exist within $v$ but do not appear in any of its child vertices. 
This setup naturally leads to the concept of \emph{vertical separators}. 
Given that $v$ corresponds to an interval $[a,b]$ along the $x$-axis, we define a vertical separator at the midpoint $(a+b)/2$. 
Any copy that spans both sides of this vertical separator is counted within $v$ but not in any of $v$’s children.

\begin{figure}[h]
    \centering
    \includegraphics[width=0.6\linewidth]{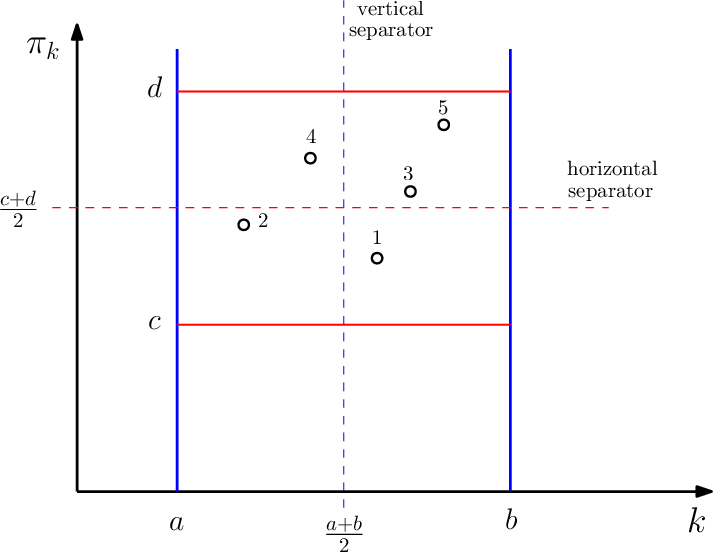}
    \caption{
        This sketch depicts the notion of vertical and horizontal global separators.
        In this example, the vertical dashed (blue) line is a vertical separator, splitting the range $[a, b]$ into two equal-sized halves.
        The horizontal dashed (red) line is a horizontal separator.
        The example also shows a $(24135)$ copy. 
        This copy is counted only if (i) the ``2'' is to the left and the ``5'' is to the right of the vertical separator, \textbf{and}, (ii) if the ``1'' is below and the ``5'' is above the horizontal separator.
    }
    \label{fig:global-separator}
\end{figure}
In addition to vertical separators, we introduce \emph{horizontal separators} that further partition each $v$ based on the $y$-axis. 
This second layer of separation divides the region into four distinct quadrants. We refer to \Cref{fig:global-separator,fig:global-separator-more-details} for an illustration.
In addition, we consider all \emph{valid configurations} of $24135$ copies relative to these quadrants. 
For instance, we can enforce that specific elements (e.g., ``2'' and ``5'') fall on opposite sides of the vertical separator and that others (e.g., ``1'' and ``5'') fall on opposite sides of the horizontal separator. 
This structure ensures that each copy of the pattern is counted exactly once within a unique configuration. Crucially, it turns out that this structure also induces monotonicity and allows for using the Birg\'e decomposition for efficient approximate counting.

\subsection{Discussion}
\label{sec:discuss}
Our results and techniques open several interesting follow-up questions, both as the first approximate permutation pattern counting results for $k > 2$ and due to the novel use of the Birg\'e decomposition. We discuss these open questions here. Along the way, we discuss (and compare our approach to) other approximate counting methods from the literature.

\paragraph{The complexity of approximate counting as a function of $n$ and $k$.}
The main open question is on the complexity of approximate counting for general (small) $k$, beyond the regime $k \leq 5$ considered in this paper.
Since our techniques are tailored ad hoc to very small patterns, it is not clear how to apply them in a more general setting, and new ingredients are likely required.

As discussed, there are complexity separations between detection and exact counting of permutation patterns: detection takes $O(n)$ time for any fixed length $k$, while exact counting requires $n^{1+\Omega(1)}$ time for $k=4$, assuming Strong 3-SUM, and $n^{\Omega(k / \log k)}$ time when $k$ is a parameter, assuming ETH.
Approximate counting lies between exact counting and detection, and it is a priori unclear where its complexity sits between linear in $n$ (for detection) and nearly worst-possible (for exact counting).

\begin{question}[Complexity of approximate counting]\label{q:approx}
What is the time complexity of \emph{approximating} the number of $\sigma$-patterns in an input sequence $f \colon [n] \to \R$ to within a $(1+\epsilon)$-multiplicative error, as a function of $n$ and $k = |\sigma|$?
\end{question}

Establishing tight upper and lower bounds for \Cref{q:approx} appears to be challenging. 
Even for exact pattern counting, a more extensively studied problem, there remains a gap between the best known upper bound of $n^{k/4 + o(k)}$ and the conditional lower bound of $n^{\Omega(k / \log k)}$, both attained by Berendsohn, Kozma and Marx \cite{berendsohn2021finding}. 
Nevertheless, given the separation we establish for $k=4$ and $k=5$ (along with the new techniques which are specially suited for approximate computation)
it is tempting to conjecture that the complexity of approximate counting in the general case, as a function of $n$ and $k$, is fundamentally lower than that of exact counting. We make the following conjecture.
\begin{conjecture}
\label{conj:better_upper_bound}
The time complexity of approximate counting $\sigma$-copies in a length-$n$ sequence, as a function of $n$ and $k=|\sigma|$, is asymptotically smaller than that of exact counting for the same parameters.
\end{conjecture}
Proving any bound of the form $n^{o(k/\log k)}$ would affirm this conjecture. 
But even improving upon the state of the art for exact counting would be interesting. 
The current best known approach of \cite{berendsohn2021finding} formulates the pattern matching instance as a constraint satisfaction problem (CSP) with binary constraints. The complexity of solving this CSP is $O(n^{t+1})$, where $t$ is the treewidth of the incidence graph of the pattern $\pi$ (see also the work of Ahal and Rabinovich \cite{AhalRabinovich2008} for an earlier investigation of the role of treewidth in this context). 
The basic constraint graph has treewidth bounded by $k/3 + o(k)$; Berendsohn et al.~combine the tree-width based approach with a gridding technique based on ideas of Cygan, Kowalik, and Soca\l{}a~\cite{Cygan2019} to reduce the exponent to $k/4 + o(k)$.

Going back to the case where $k$ is a small fixed constant, in parallel to our work, Beniamini and Lavee \cite{BL25} very recently extended the line of work on exact small pattern counting to slightly larger patterns, of length up to $k=7$. Using exact multidimensional pattern trees (including an exact analogue of the $12$-range counting data structure we use in this paper) they were able to obtain an $\tilde{O}(n^2)$-time exact algorithm for $k \leq 7$, and a $\tilde{O}(n^{7/4})$ exact algorithm for $k = 5$. As an interesting special case of~\Cref{conj:better_upper_bound}, it would be interesting to explore approximate counting in the same regime. Another very recent and independent work, by Diehl and Verri \cite{DiehlVerri2025}, extends the (exact) corner trees framework of \cite{EvenZoharLeng2021} to a larger set of double poset structures.

\paragraph{Stratified sampling and approximate counting to detection reductions.}
Perhaps the most generic and widely-used approach to approximate counting in the literature is a randomized technique based on stratified (non-uniform) sampling. 
An intriguing line of recent work \cite{DLM22,DLM24,CensorHillelEvenWilliams2025} (see also \cite{DL21}) develops generic \emph{approximate counting to detection} reductions, which use stratified sampling as a fundamental primitive. These frameworks require the problem to be presented as a counting or detection problem on an (implicit) $k$-partite hypergraph. Suppose one is given a hyperedge detection oracle $\mathcal{O}$ that is able to detect whether a given subhypergraph (in an input $k$-partite hypergraph) is non-empty. Using stratified sampling, the framework implements an approximate counting algorithm for $k$-partite hyperedges that makes only a \emph{polylogarithmic} number of queries to the detection oracle.

\paragraph{How to make stratified sampling work: A rainbow detection approach.}
Recall that permutation pattern detection is solvable in $O(n)$ time for any fixed pattern length $k$ \cite{GuillemotMarx2014}. Naturally, the reader may ask now whether this linear-time detection algorithm can be readily combined with the approximate counting to detection reduction framework, to obtain a near-linear time approximate counting algorithm for all fixed $k$, thereby settling the main open problem and solving~\Cref{conj:better_upper_bound} in an ultimate sense. Unfortunately, it is unlikely that the (vanilla) pattern detection problem can be cast in a $k$-partite hypergraph format. But all is not lost: Consider the following \emph{colorful permutation pattern detection} problem. In this problem, the elements of the input permutation are randomly colored by one of $k$ colors in advance, where $k$ is the pattern length as usual, and the color of each element is chosen uniformly and independently of all other colors. At the end of the coloring process, each color class corresponds to one part of a $k$-partite hypergraph, and the hyperedege detection problem corresponds to detecting a rainbow copy of the pattern in the input permutation. (For simplicity of the discussion, we skip some of the implementation details here.)
Thus, we conclude that an efficient permutation pattern detection algorithm for \emph{rainbow copies} in the randomly colored setting would immediately lead to an approximate counting algorithm (in the standard setting, without colors) with the same running time, up to lower order terms.

Next, one may ask whether the analysis of Guillemot and Marx from \cite{GuillemotMarx2014} can be applied to the rainbow setting. Roughly speaking, the main technical result in their paper is a ``structure versus chaos'' type result, which in modern language asserts that either the (twin-)width of a certain structure is bounded by a constant, or the input permutation satisfies the conditions of the Marcus-Tardos theorem \cite{MarcusTardos2004}, thus containing a $k \times k$ subgrid which itself contains a copy of the desired pattern.\footnote{The Marcus-Tardos theorem \cite{MarcusTardos2004} in extremal combinatorics asserts that for fixed $k$ and large $n$, any subset of at least $C(k) \cdot n$ points in the $n \times n$ grid contains an induced subgrid of size $k \times k$, where $C(k)$ depends only on $k$ (and not on $n$).}
Thus, in order to adapt the Guillemot-Marx algorithm to the rainbow setting, one would need at the very least to prove a suitable rainbow version of the Marcus-Tardos theorem, if true at all. Here is one possible version of interest.
\begin{question}[A possible rainbow version of the Marcus-Tardos theorem]
Let $S$ be a set of points in the $n \times n$ grid, where each point in $S$ is colored randomly and independently in one of $k$ colors (the same set of colors for all points). Is it true that there exists a constant $C(k)$ depending only on $k$ (and not on $n$), such that if $S$ contains at least $C(k) \cdot n$ points, then with high probability, $S$ contains a $k \times k$ subgrid where each column is monochromatic, and each row is rainbow?
\end{question}

We complete the discussion by mentioning that the approximate counting techniques we employ in our paper can be easily translated into rainbow detection techniques. Thus, a slight modification of our techniques yields rainbow detection in near-linear time for all patterns of length $k \leq 5$. Through the use of the approximate-counting-to-rainbow-detection framework, one can obtain an alternative proof to our main results. However,
\begin{itemize}
\item This proof technique can only give a \emph{randomized algorithm} (while our algorithm is deterministic); and
\item Such a proof would not be materially simpler than the one in our paper. The proof would still go through enumerating over a number of different techniques, though the techniques themselves slightly differ (both Birg\'e approximation and the approximate $12$-counting data structure should be replaced by rainbow detection variants, while the separator-based ideas remain unchanged).
\end{itemize}

\paragraph{Relation to parameterized width notions.}
Another possible avenue to proving better approximate upper bounds is through studying connections to width notions from the parameterized complexity literature. As we saw, algorithmic results for both detection and exact counting make use of such notions: the former gave rise to twin-width \cite{GuillemotMarx2014,BonnetTwinWidth2021} and the latter makes heavy use of tree-width \cite{AhalRabinovich2008,berendsohn2021finding}. It would be very intriguing to explore what role such width notions may play in the approximate version of pattern counting. The fact that approximate counting (in the small $k$ case) admits techniques that go beyond the exact case may suggest that either a complexity notion other than tree-width is at play here, or we can use the new techniques to bound the tree-width of an easier subproblem (with more of the values constrained due to the use of, say, substructure monotonicity and Birg\'e approximation).

\paragraph{Toward super-linear approximate lower bound.}
From the lower bound side, essentially no nontrivial (superlinear) results are known for the Word RAM model, and proving any $\omega(n)$ lower bound that applies to the approximate counting of some fixed-length patterns would be interesting. 

\begin{conjecture}
There exists a pattern $\sigma$ of constant length for which approximate counting of $\sigma$ in length-$n$ sequences requires $\omega(n)$ time in the Word RAM model.
\end{conjecture}

For $k=3,4,5$, the existing algorithms for, say, 2-approximate counting (and exact counting, for $k=3$) have time complexity $n \log^{O(1)} n$. This raises the question of whether the polylogarithmic dependence is necessary (for $k=2$ it is \emph{not} necessary \cite{ChanPatrascu2010}). We conjecture that the answer is positive already for $k=4$. 

To this end, note that even if the problem of colorful detection described above admits an $O(n)$ time algorithm, this does not rule out a (slightly) super-linear approximate counting lower bound, since known approximate counting to detection reductions make a polylogarithmic number of calls to the detection oracle, and thus will only yield an upper bound of order $n \cdot \log^{O(1)}(n)$ at best.

\paragraph{Birg\'e approximation in other contexts.}
Finally, the use of Birg\'e decomposition in this paper seems to be novel in the context of pattern counting and, perhaps more generally, in combinatorial contexts beyond the scope of distribution testing. 
This decomposition is very useful in our setting as many sequences of quantities turn out to be monotone. It would be interesting to find other counting problems in low-dimensional geometric settings which exploit monotonicity in non-trivial ways, and/or explore how Birg\'e approximation relates to existing approximate counting techniques in the computational geometry literature. One such technique known to be implementable deterministically is that of (standard and shallow) \emph{cuttings}, see \cite{chan2016optimal} and the references within.

\section{Preliminaries}


\subsection{Segment trees}
\label{sec:segment-trees}

We use a natural and standard representation of permutations, in which a permutation $\pi$ is represented by a set of points $\{(i, \pi_i) : i \in [n]\}$ in plane.
On this set of points, our algorithms perform simple counting queries. 

\begin{lemma}[Segment tree data structure]
\label{lemma:segment-tree}
Let $\pi$ be a permutation over $[n]$. Define 
$$
 S_{i,j}^{a,b} := |\{ x \in [n] : i \leq x \leq j, a \leq \pi(x) \leq b\}|.
$$
and $N_{i,j}^{a,b} = |S_{i,j}^{a,b}|$.
There exists a data structure that, given $\pi$, initializes in time $O(n \log^2 n)$ using $O(n \log n)$ space, and supports the following operations in time $O(\log^2 n)$: 
\begin{enumerate}[1.]
\item Value and location counts: given indices $i \leq j \in [n]$ and values $a \leq b \in [n]$, return $N_{i,j}^{a,b}$.
\item Query access to locations in segment: Given $i,j,a,b$ as above, and $1 \leq \ell \leq N_{i,j}^{a,b}$, return the index of the $\ell$-th leftmost element within the set $S_{i,j}^{a,b}$.

\item Query access to values in segment: Given $i,j,a,b$ as above and $1 \leq \ell \leq N_{i,j}^{a,b}$, return the $\ell$-th largest value within the set $\{\pi(x) : x \in S_{i,j}^{a,b}\}$.
\end{enumerate}

\end{lemma}
\Cref{lemma:segment-tree} can be obtained using standard techniques in the data-structure literature. For completeness, we outline these techniques in \Cref{sec:2D-segment-tree}.

\subsection{Birg\'e decomposition: Coresets for monotone functions}
\label{sec:birge}

\review{At a first glance the permutation pattern counting and matching problem seem to have nothing to do with monotone functions (as long as neither the input function nor the permutation pattern are monotone). And yet, tools for approximating sums of discrete (non-negative) monotone functions play a central role in our algorithms.}

\review{For our purposes we will use the following result of Har-Peled, on coresets for monotone functions \cite{HarPeled2006}. Results of this type are often called Birg\'e decompositions in statistics and distribution testing \cite{Birge1987, DDSVV13, DDS14}, and have been studied in additional contexts, see, e.g., the notion of $k$-approximation functions studied by Halman et al.~\cite{Halman2014}.}

\begin{lemma}[Fast approximation of monotone sums \cite{HarPeled2006}]\label{lem:Birge}
\review{Let $0 < \eps < 1$ and $n \in \N$ be known parameters, and suppose we are given query access to a monotone sequence $x_1 \geq x_2 \geq \ldots \geq x_n \geq 0$ of real numbers. Then there exists a deterministic algorithm which returns a value $y \in (1 \pm \eps) \sum_{i=1}^{n} x_i$, with query complexity and running time $O(\eps^{-1} \log n )$.}

\review{Moreover, if the query access provides a multiplicative $1\pm \gamma$ approximation, then this algorithm returns a value $y \in (1\pm \gamma) (1 \pm \eps) \sum_{i = 1}^n x_i$; the algorithm is oblivious to the value of $\gamma$.}
\end{lemma}

\review{We note that the statement proved in \cite{HarPeled2006} is the one appearing in the first paragraph of the above lemma statement. The result in the second paragraph is an immediate corollary. More precisely, the algorithm from \cite{HarPeled2006} returns a weighted sum $\sum_{j=1}^{t} w_j x_{i_j}$ which is a $(1 \pm \eps)$-approximation to $\sum_{i=1}^{n} x_i$, where $t = O(\log n / \eps)$ and both the weights $w_j$ and the indices $i_j$ are chosen deterministically and do not depend on values in the sequence. Replacing each summand $x_{i_j}$ with a $(1 \pm \gamma)$-approximation of it can thus only change the value of the weighted sum by up to a multiplicative $1 \pm \gamma$ factor.}

\section{Almost All $4$-Patterns via Birg\'e Decomposition}
\label{sec:almost-all-4-length}
\subsection{Symmetry Reductions and Known Results}
It is folklore that patterns form certain equivalency groups. For instance, counting $1234$ copies is equivalent to counting $4321$: the number of $4321$ copies in a sequence $\pi$ is the same as that of $1234$ copies in the reverse of $\pi$ (i.e., in $\pi$ reflected along a vertical axis). Similarly, counting $2134$ copies is equivalent to counting $3421$, as the number of $3421$ copies in a sequence $\pi$ is the same as the number of $2134$ copies in the sequence in which the $i$-th element equals $n+1 - \pi_i$ (i.e., in the sequence reflected along a horizontal axis).

Moreover, it has been established that copies of some of these $8$ patterns can be counted in near-linear time, even exactly.
\begin{theorem}[\cite{EvenZoharLeng2021}]
    The occurrences of copies $1234$, $1243$, and $2143$ in a given permutation can be counted in time $\tO(n)$.
\end{theorem}
In this work, our goal is to show that the remaining $4$-length copies can be \emph{approximately} counted also in near-linear time. 
It turns out that the remaining non-symmetric ones, i.e., $1324$, $1342$, $1423$, $1432$, and $2413$, can be split into two categories based on their properties.
In this section, we describe an idea that enables us to approximately count the number of copies of $1324$, $1342$, $1423$, and $1432$. To handle $2413$, we in detail introduce an additional idea in \Cref{sec:2413}.

\subsection{Approximately Counting $1324$ Copies}
In this section, we describe how to approximately count $1324$ copies and then extend that to $1342, 1423$, and $1432$; see \Cref{sec:approx-counting-other-copies}.
Fix a permutation $\pi$, let $\cC$ be all $1324$ copies in $\pi$. The underlying idea of our approach is to partition $\cC$ into classes so that it is relatively easy to approximate the size of each class.

\subsubsection{Fixing ``$3$'' in $1324$ copies}
As the first step, our approach partitions $\cC$ with respect to the value of the ``$3$''. In particular, there are $n$ such classes $\cC_1, \cC_2, \ldots, \cC_n$ such that $\cC_i$ is the subset of $\cC$ with their $3$-value being equal $\pi_i$. 
Clearly,
\[
    |\cC| = \sum_{i = 1}^n |\cC_i|.
\]
Our approach approximates each $|\cC_i|$ independently.
The main technical contribution of our work is showing that $\cC_i$ can be further partitioned into classes that exhibit certain monotonicity in their size. Our approach employs the Birg\'e decomposition (\Cref{lem:Birge}) to leverage that property and approximate $|\cC_i|$ in only $\poly \log n$ time. We now describe the details of this idea.

\subsubsection{Monotonicity with respect to ``$4$'' within $\cC_i$}
Consider the example in which $\pi = 136548279$, as illustrated in \Cref{fig:fix-3}.
Fix $6$ to be ``3'' in a $1324$ copy. 
In that case, the candidates for ``4'' are $8, 7$, and $9$. Since $7$ and $9$ appear in $\pi$ after $8$, any $1324$ copy of the form $\pi_x 6\pi_y 8$ also yields $1324$ copies $\pi_x 6\pi_y 7$ and $\pi_x 6\pi_y 9$. We formalize this observation as follows.

\begin{lemma}
\label{lemma:1324-monotone-wrt-4}
    Let $\cC_i$ be the set of all $1324$ copies of a permutation $\pi$ such that ``3'' equals $\pi_i$. 
    Let $\cC_{i, j}$ be the set of all $1324$ copies such that $\pi_x \pi_i \pi_y \pi_j \in \cC_{i}$. Then, $|\cC_{i, j}| \le |\cC_{i, j'}|$ for each $i < j < j'$ with $\pi_i < \pi_j, \pi_{j'}$.
\end{lemma}
\begin{proof}
    Let $\pi_x \pi_i \pi_y \pi_j \in \cC_{i, j}$. By the fact that we consider $1324$ copies, it implies $\pi_x < \pi_i$, $\pi_y < \pi_i$ and $y < j$. Since $\pi_i < \pi_j$, $\pi_i < \pi_{j'}$ and $j < j'$, we have that $\pi_x \pi_i \pi_y \pi_{j'} \in \cC_{i, j'}$.
\end{proof}
\Cref{lemma:1324-monotone-wrt-4} essentially states that $\cC_i$ can be partitioned into subsets $\cC_{i, j}$ whose sizes are non-decreasing when ordered with respect to $j$ -- the claim excludes $\cC_{i, j}$ for which $\pi_j \le \pi_i$; for such cases $\cC_{i, j} = \emptyset$ anyway.
By \Cref{lem:Birge}, this further implies that to approximate $|\cC_i|$, it suffices to (approximately) compute $|\cC_{i, j}|$ for only $O(\log n / \eps)$ different values of $j$ with $\pi_j > \pi_i$ and $j > i$.
Our next goal is to discuss how to compute an approximation of $|\cC_{i, j}|$.

\subsubsection{Monotonicity with respect to ``$2$'' within $\cC_{i, j}$}
Recall that $\cC_{i, j}$ corresponds to all $1324$ copies with ``3'' being $\pi_i$ and ``4'' being $\pi_j$. 
As before, consider the example in which $\pi = 136548279$, as illustrated in \Cref{fig:fix-3-and-4}.
Fix $6$ to be ``3'' and $9$ to be ``4'' in a $1324$ copy. 
In that case, the candidates for ``2'' are $2, 4$, and $5$. 
Since $4 > 2$ and $5 > 2$, any $1324$ copy of the form $\pi_x 6 2 9$ also yields $1324$ copies $\pi_x 6 4 9$ and $\pi_x 6 5 9$. We formalize this observation as follows.
\begin{figure}
    \centering
    \includegraphics[width=0.4\linewidth]{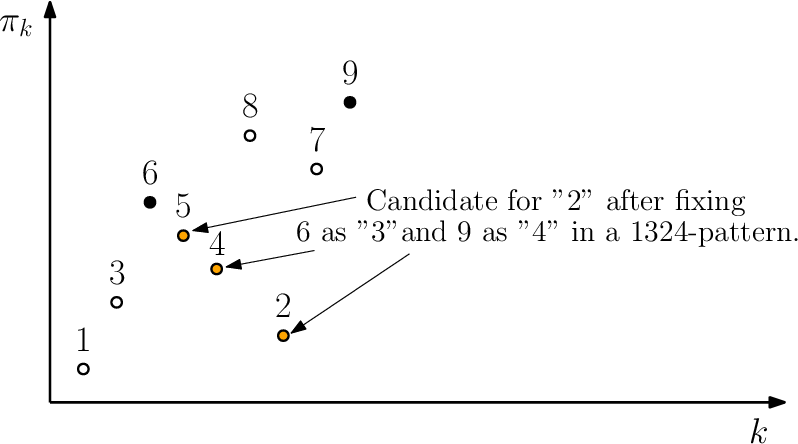}
    \caption{The illustration corresponds to permutation $\pi = 136548279$, depicted in a plane at points $(i, \pi_i)$.}
    \label{fig:fix-3-and-4}
\end{figure}
\begin{lemma}
\label{lemma:1324-monotone-wrt-4-and-2}
    Let $\cC_{i.j}$ be the set of all $1324$ copies of a permutation $\pi$ such that ``3'' equals $\pi_i$ and ``4'' equals $\pi_j$. 
    Let $\cC_{i, j, k}$ be the set of all $1324$ copies such that $\pi_x \pi_i \pi_k \pi_j \in \cC_{i, j}$.
    Then, $|\cC_{i, j, k}| \le |\cC_{i, j, k'}|$ for each $\pi_k < \pi_{k'}$ with $i < k < j$, $i < k' < j$ and $\pi_k, \pi_{k'} < \pi_i$.
\end{lemma}
\begin{proof}
    Let $\pi_x \pi_i \pi_k\pi_j \in \cC_{i, j, k}$. 
    By the fact that we consider $1324$ copies, it implies $\pi_x < \pi_k$ and $x < i < k$. Since $\pi_k < \pi_{k'} < \pi_i$ and $i < k' < j$, we have that $\pi_x \pi_i \pi_{k'} \pi_j \in \cC_{i, j, k'}$.
\end{proof}
\Cref{lemma:1324-monotone-wrt-4-and-2} states that $\cC_{i, j}$ can be partitioned into subsets $\cC_{i, j, k}$ whose sizes are non-decreasing when ordered with respect to $\pi_k$.
By \Cref{lem:Birge}, this further implies that to approximate $|\cC_{i,j}|$, it suffices to (approximately) compute $|\cC_{i, j, k}|$ for only $O(\log n / \eps)$ different values of $k$ with $i < k < j$ and $\pi_k < \pi_i$.
Coupling this with \Cref{lemma:1324-monotone-wrt-4}, $|\cC_{i}|$ can be approximate by computing $|\cC_{i, j, k}|$ for only $\poly(\log{(n)} / \eps)$ different pairs of $j$ and $k$.

\subsubsection{Algorithm}
As a reminder, $\cC_{i, j, k}$ is the set of all $1324$ copies such that ``3'' equals $\pi_i$, ``4'' equals $\pi_j$, and ``2'' equals $\pi_k$. 
$|\cC_{i, j, k}|$ is computed by counting the number of points $(\ell, \pi_\ell)$ such that $1 \le \ell \le i - 1$ and $1 \le \pi_{\ell} < \pi_k - 1$. This can be done in $\poly \log n$ time using sparse segment trees, as provided by \Cref{lemma:segment-tree}. This now enables us to provide the pseudo-code of our approach (\Cref{alg:approximate-1324}).

\begin{algorithm}
\begin{algorithmic}[1]
\medskip 
\Statex \textbf{Input:} A permutation $\pi$; an approximation parameter $\eps > 0$
\Statex \textbf{Output:} A $1+\eps$ approximation of the number of $1324$ copies in $\pi$
\medskip
\Statex \hrule 

    \State Build a sparse segment tree $S$ on $(i, \pi_i)$ for all $i = 1 \ldots n$
    \For{$i = 1 \ldots n$}\Comment{Fix ``3''}
        \State Let $J$ be the set of candidates for ``4'' in 1324 copies given that ``3'' is fixed to $\pi_i$.
        \State Let $J'$ be the subset of $J$ queried by the algorithm in \Cref{lem:Birge}.
        \For{$j \in J'$}\Comment{Fix ``4''}
            \State Let $K$ be the set of candidates for ``2'' in 1324 copies given that ``3'' is fixed to $\pi_i$ and ``4'' is fixed to $\pi_j$.
            \State Let $K'$ be the subset of $K$ queried by the algorithm in \Cref{lem:Birge}.
            \For{$k \in K'$}\Comment{Fix ``2''}
                \State Let $c_{i, j, k}$ be the number of points $(x, \pi_x)$ in $S$ such that $x \le i - 1$ and $\pi_x \le \pi_k - 1$.\label{line:compute-c'}
            \EndFor
            \State Use the algorithm from \Cref{lem:Birge} to output a $1\pm \eps/3$ approximation of $|\cC_{i,j}|$ by using $c_{i, j, k}$ as the query points. Denote that approximation by $c_{i, j}$.
        \EndFor
        \State Use the algorithm from \Cref{lem:Birge} to output a $1\pm \eps/3$ approximation of $|\cC_{i}|$ by using $c_{i, j}$ as the query points. Denote that approximation by $c_{i}$.
    \EndFor
    \State \Return $\sum_{i = 1}^n c_i$
    \end{algorithmic}
\caption{\textsc{Approximate-1324-Copies}}
\label{alg:approximate-1324}
\end{algorithm}

We are now ready to show the following.
\begin{theorem}
    Given a permutation $\pi$ and an approximation parameter $\eps \in (0, 1)$, \Cref{alg:approximate-1324} computes a $1 \pm \eps$ approximation of the number of $1324$ copies in $\pi$ in time $O(n \cdot \poly(\log(n) / \eps))$.
\end{theorem}
\begin{proof}
    We analyze separately the running time and the approximation guarantee.
    \paragraph{Running time.}
    There are $n$ options to choose $i$. By \Cref{lem:Birge}, $|J'|, |K'| \in O(\log(n) / \eps)$. 
    Note that the sets $J$ and $K$ need not be constructed explicitly. It suffices to, for a given $t$, be able to access the $t$-th element of those sets, which can be done in $O(\log^2 n)$ time \review{using query access to locations and values in the segment tree $S$; the necessary operations are guaranteed by \Cref{lemma:segment-tree}.}
    Finally, \Cref{line:compute-c'} of \Cref{alg:approximate-1324} can be executed in $O(\log^2 n)$ time; see \Cref{lemma:segment-tree}.

    Therefore, the overall running time is $O(n \cdot \poly(\log(n) / \eps))$.

    \paragraph{Approximation guarantee.}
        Let $c_{i, j,k}$, $c_{i, j}$ and $c_i$ be as defined in \Cref{alg:approximate-1324}. Observe that $c_{i, j, k} = |\cC_{i, j, k}|$.
        By the guarantee of the algorithm in \Cref{lem:Birge}, we have $c_{i, j} \in (1 \pm \eps/3) \cdot |\cC_{i, j}|$.

        Since $c_{i, j}$ are used to obtain an approximation $c_i$ of $|\cC_{i}|$, by \Cref{lem:Birge} we have that $c_i \in (1 \pm \eps/3) \cdot (1\pm \eps/3) \cdot |\cC_i| \in (1 \pm \eps) \cdot |\cC_i|$, for $\eps \in (0, 1)$.
\end{proof}

\subsection{Approximately Counting $1342, 1423$ and $1432$ Copies}
\label{sec:approx-counting-other-copies}
\review{Remark: For the sake of readability, below we use ``increasing'' in place of ``non-decreasing'' and, likewise, ``decreasing'' in place of ``non-increasing''.}
\\
Our algorithm to approximately count $1324$ copies can be described as follows: Fix ``3''; then, ``4'' counts are position-increasing; then, ``2'' counts are value-increasing. 
In the same way can be described the algorithms for approximately counting copies of $1342, 1423$, and $1432$. 
We provide those descriptions below, and the formal proofs follow exactly the same lines as for $1324$ copies.
\begin{enumerate}
    \item[1342:] Fix ``3''; then, ``4'' counts are position-decreasing; then, ``2'' counts are value-increasing.
    \item[1423:] Fix ``2''; then, ``3'' counts are value-decreasing; then, ``4'' counts are position-increasing.
    \item[1432:] Fix ``3''; then, ``2'' counts are value-increasing; then, ``4'' counts are position-increasing.
\end{enumerate}
\section{Remaining $4$-Patterns via Birg\'e and Separators}
\label{sec:2413}
The main idea behind approximate counting of $1324$ copies was to fix one of the positions and then show that the counts are monotone with respect to two other positions, e.g., fix ``3'', then the counts are monotone with respect to the position of ``4''; after fixing ``3'' and ``4'', the counts are monotone with respect to the value of ``2''. 
Unfortunately, copies of $2413$ do not seem to exhibit such a property.
To alleviate that, we observe that there is an additional way of partitioning the copies of $2413$.

To illustrate this partitioning approach, assume that we fix ``4''. Then, we would like to exhibit the monotonicity of the copy counts with respect to the value or position of at least one among ``2'', ``1'', and ``3''. However, this is not the case. 
Intuitively, the challenge here is that the tools we developed so far do not enable us to approximate the number of copies of $12$ in a given permutation in $\poly(\log n, 1/\eps)$ time.
To see how it affects counting $2413$ copies, for instance, after fixing a ``4'', no special structure is imposed on the candidates of ``1'' and ``3''!
Indeed, even though both ``1'' and ``3'' have to be to the right and below the fixed ``4'', our algorithm still needs to (approximately) count the number of monotone pairs in a given subarray.

What if we are concerned \emph{only} with the number of $2413$ copies in which the position of ``1'' is less than $s$, while the position of ``3'' is greater than $s$? This situation is illustrated in \Cref{fig:separator-example}, and $s$ should be thought of as ``separator''.
After imposing this additional structure between ``1'' and ``3'', the counts become monotone with respect to the value of ``3''. 
Hence, we can again apply the Birg\'e theorem for approximating the counts.

It remains to show that there exists a small number of separators that enable counting all $2413$ copies. We dive into those details in the rest of this section, describing how to partition ``3'' and ``4'' into certain buckets that allow for the described $2413$-copy partitioning.
Ultimately, this section leads to the following result:
\begin{theorem}[Approximating $2413$ copies]\label{thm:2413-main}
There exists a deterministic algorithm for approximating the number of $2413$ copies in a permutation of length $n$ to within a multiplicative factor of $1+\eps$, with running time of $n \cdot \poly(\log n, 1/\eps)$.
\end{theorem}

\paragraph{Organization of this section.}
We begin by, in \Cref{sec:2413-preliminaries}, stating several definitions that are instrumental in describing our partitioning of $2413$.
\Cref{sec:main-theorem-2413} outlines our proof of \Cref{thm:2413-main}, while \Cref{sec:approximating-4-heavy,sec:approximating-3-heavy} prove the main technical claims we need in the proof of \Cref{thm:2413-main}.

\subsection{Preliminaries}
\label{sec:2413-preliminaries}
For convenience, we let $[n] := \{0,1,\ldots,n-1\}$.
We begin by defining the notion of $j$-buckets and type-$j$ copies, which are instrumental in defining the kind of separator we use and illustrate in \Cref{fig:separator-example}. Recall that a copy of $2413$ in a permutation $\pi \colon [n] \to [n]$ is any quadruple of indices $i_1 < i_2 < i_3 < i_4$ such that $\pi(i_3) < \pi(i_1) < \pi(i_4) < \pi(i_2)$.

\begin{definition}[Type of copy; $j$-buckets]
\label{definition:copy-buckets}

For each index $i \in [n]$ consider the standard binary representation of $i$ using $\lceil \log n \rceil$ bits, and define the $j$-least significant bit (or $j$-LSB in short) as the term corresponding to $2^j$ in the binary representation. We say that a $2413$ copy $(i_1, i_2, i_3, i_4)$ in $\pi$ is type-$j$ if $i_2$, i.e., the index of the ``4'', and $i_4$ (the index of the ``3'') differ on the $j$-LSB, but have equal $j'$-LSB for all $j' > j$.

Finally, two indices in $[n]$ are said to be in the same $j$-bucket if their $j'$-LSB is equal for all $j' \geq j$.
This definition is illustrated in \Cref{fig:type-j}.
\end{definition}
Observe that there are many $j$-buckets. In fact, $j$-buckets partition the integers into sets of $2^j$ consecutive integers each. 
For instance, the ranges of integers $[0, 7]$, $[8, 15]$, $[16, 23]$, $[24, 31]$ are all $3$-buckets.
\review{Another perspective one can take on $j$-buckets is as follows: consider a complete binary tree of depth $\lceil \log n \rceil$ with leaves numbered $0$ through $2^{\left\lceil \log n \right\rceil}-1$ ``from left to right''. Then, two leaves $a$ and $b$ are in the same $j$-bucket if their lowest common ancestor is at distance $j$ from them.}

\begin{figure}
    \centering
    \includegraphics[width=0.5\linewidth]{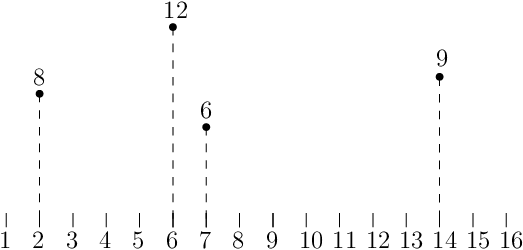}
    \caption{
    This example depicts a copy of $2413$ equal to $(2, 6, 7, 14)$ with $\pi(i_1) = 8$, $\pi(i_2) = 12$, $\pi(i_3) = 6$ and $\pi(i_4) = 9$. 
    Since $i_2 = (00110)_2$ and $i_4 = (01110)_2$, this copy is  $3$-type. Moreover, we have that all $i_1, i_2, i_3$ and $i_4$ are in the same $4$-bucket. The indices $i_1, i_2$ and $i_3$ are in the same $3$-bucket as well, while $i_2$ and $i_3$ are in addition in the same $2$-type and $1$-type bucket; see \Cref{definition:copy-buckets}.
    }
    \label{fig:type-j}
\end{figure}

Note that a bucket consists of contiguous subintervals of $[n]$. Moreover, in a $2413$ copy which is type-$j$, the ``4'' and ``2'' are in the same $(j+1)$-bucket and in different, but neighboring, $j$-buckets. This motivates the following definition.

\begin{definition}[$4$-heavy, $3$-heavy] 
Consider a type-$j$ $2413$ copy $(i_1, i_2, i_3, i_4)$. 
We say that the copy is $4$-heavy if $i_2$ and $i_3$, i.e., the ``4'' and ``1'', are in the same $j$-bucket. Otherwise, we say that the copy is $3$-heavy.
\end{definition}
Note that in a type-$j$ copy $(i_1, i_2, i_3, i_4)$ that is $3$-heavy, $i_3$ (``the $1$-entry'') is in the same $j$-bucket as $i_4$ (``the $3$-entry''). 
Similarly, in a type-$j$ copy $(i_1, i_2, i_3, i_4)$ that is $4$-heavy, $i_3$ (``the $1$-entry'') is in the same $j$-bucket as $i_2$ (``the $4$-entry''). This yields the following observation.
\begin{observation}
    Each type-$j$ copy is either $3$- or $4$-heavy, but not both.
\end{observation}

\subsection{Proof of Main Theorem}
\label{sec:main-theorem-2413}
The proof of the main result of this section, i.e., \Cref{thm:2413-main}, relies on the following two claims, saying that the count of $3$- and $4$-heavy copies of a \emph{fixed} type \review{and one location} can be approximated in poly-logarithmic time.

\begin{lemma}[Approximation of $4$-heavy copies] \label{lem:4-heavy-approx}
Let $n \in \mathbb{N}$, $i \in [n]$, and $j \in \left[ \left \lceil \log n \right \rceil \right]$. 
Let $S$ be a pre-built segment tree for $\{(i, \pi_i) : i \in [n]\}$.
There exists a deterministic algorithm with running time $\poly(\log n, 1/\eps)$ that, given access to $S$, returns a $(1+\eps)$-approximation of the number of $4$-heavy type-$j$ copies $(i_1, i_2, i_3, i_4)$ of $(2413)$ in $\pi$ for which $i_2 = i$.
\end{lemma}

\begin{lemma}[Approximation of $3$-heavy copies] \label{lem:3-heavy-approx}
Let $n \in \mathbb{N}$, $i \in [n]$, and $j \in \left[ \left \lceil \log n \right \rceil \right]$. 
Let $S$ be a pre-built segment tree for $\{(i, \pi_i) : i \in [n]\}$.
There exists a deterministic algorithm with running time $\poly(\log n, 1/\eps)$ that, given access to $S$, returns a $(1+\eps)$-approximation of the number of $3$-heavy type-$j$ copies $(i_1, i_2, i_3, i_4)$ of $(2413)$ in $\pi$ for which $i_4 = i$.
\end{lemma}

With \Cref{lem:4-heavy-approx,lem:3-heavy-approx} in hand, the proof of \Cref{thm:2413-main} is almost immediate.
\begin{proof}[Proof of \Cref{thm:2413-main}] 
Our main algorithm is given as \Cref{alg:approximate-2413}.

\begin{algorithm}
\begin{algorithmic}[1]
\medskip 
\Statex \textbf{Input:} A permutation $\pi$; an approximation parameter $\eps > 0$
\Statex \textbf{Output:} A $1+\eps$ approximation of the number of $2413$ copies in $\pi$
\medskip
\Statex \hrule 

    \State Build a segment tree $S$ on $(i, \pi_i)$ for all $i = 1 \ldots n$ as described in \Cref{sec:2D-segment-tree}.
    \For{$i = 0 \ldots n-1$}
        \For{$j = 0 \ldots \lceil \log n \rceil$}
            \State Compute the approximate count $C_{j,4,i}$ of type-$j$ $4$-heavy copies of $(2413)$ whose $4$-entry is at location $i$, using the algorithm of \Cref{lem:4-heavy-approx} and passing $S$ to it. 
            \State Compute the approximate count $C_{j,3,i}$ of type $j$ $3$-heavy copies of $(2413)$ whose $3$-entry is at location $i$, using the algorithm of \Cref{lem:3-heavy-approx} and passing $S$ to it. 
        \EndFor
    \EndFor
    \State \Return $\sum_{i,j} (C_{j,3,i} + C_{j,4,i})$
    \end{algorithmic}
\caption{\textsc{Approximate-2413-Copies}}
\label{alg:approximate-2413}
\end{algorithm}

Each copy of $(2413)$ in $\pi$ is type-$j$ for exactly one value of $j$, and moreover, each such copy is either $3$-heavy or $4$-heavy, but not both. 
Hence, the sum $\sum_{i,j} (C_{j,3,i} + C_{j,4,i})$ is a $(1+\eps)$-approximation of the number of $(2413)$-copies in $\pi$.

Since $S$ can be built in $\tO(n)$ time, and each invocation to the algorithms from \Cref{lem:4-heavy-approx,lem:3-heavy-approx} takes $\poly(\log n, 1 / \eps)$ time, \Cref{alg:approximate-2413} runs in $\tO\rb{n \cdot \poly(1 / \eps)}$ time.
\end{proof}

It remains to prove \Cref{lem:4-heavy-approx} and \Cref{lem:3-heavy-approx}. We refer the reader to \Cref{fig:2413-heavy-4} for an illustration of the $4$-heavy case (\Cref{lem:4-heavy-approx}).

\subsection{Approximating the Number of $4$-Heavy Copies}
\label{sec:approximating-4-heavy}

\begin{figure}
    \centering
    \includegraphics[width=0.8\linewidth]{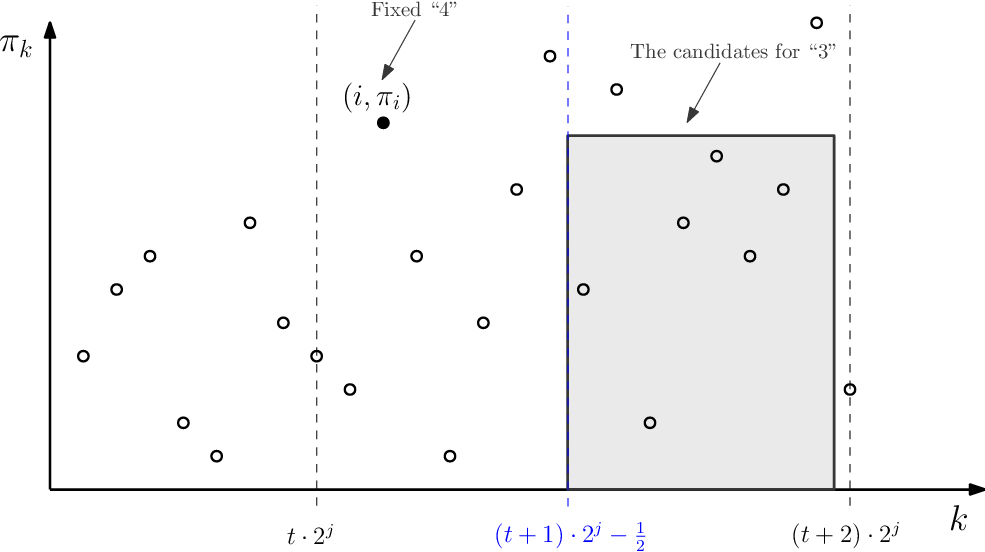}
    \caption{A helper illustration for the proof of \Cref{lem:4-heavy-approx}.
    In this sketch, $t$ is an integer.
    The shaded rectangle corresponds to the set $X$ of ``candidates for $3$''.
    }
    \label{fig:2413-heavy-4}
\end{figure}

\begin{proof}[Proof of \Cref{lem:4-heavy-approx}]
Fix $j$. Consider the set $A_i$ of all $4$-heavy type-$j$ $(2413)$-copies $(i_1, i_2, i_3, i_4)$ where $i_2 = i$, i.e., all copies where the ``4'' is located at index $i$. 
Recall that these are precisely all copies where the ``1'' is in the same $j$-bucket as index $i$, whereas the ``3'' is not in the same bucket. 
In particular, for any two such copies $(i_1, i, i_3, i_4)$ and $(i'_1, i, i'_3, i'_4)$, it holds that $i'_4 > i_3$, i.e., \textbf{all} ``3's'' lie to the right of \textbf{all} ``1's''.\footnote{Throughout our proofs, it is instructive to picture the input as a set of points with coordinates $(i, \pi_i)$ for all $i \in [n]$. 
The terminology such as ``left'', ``right'', ``above'', and ``below'' is defined with respect to that depiction of the input.}

\paragraph{``3'' candidates.}
We approximate $|A_i|$ by first fixing the ``candidate for $3$''; the ``candidate for $4$'' is already fixed by the definition of $A_i$. 
Moreover, we show that a particular function is monotone with respect to those candidates, which will enable us to apply \Cref{lem:Birge}; we invoke \Cref{lem:Birge} with parameter $\eps / 3$.
Formally, let $X$ be the set of indices with the following two properties: (1) the indices in the $j$-bucket immediately neighboring to ``the right'' the $j$-bucket $i$ belongs to, and (2) the indices whose value is smaller than $\pi(i)$. In \Cref{fig:2413-heavy-4}, $X$ corresponds to the shaded area.

For each such candidate $x \in X$, define $f(x)$ as the number of $(2413)$-copies in $A_i$ whose ``3'' corresponds to $x$. 
Importantly, $f(x)$ is ``monotone by value'' within the relevant bucket. Precisely, within $A_i$, $f(x)$ is non-decreasing as a function of $\pi(x)$ over $x \in X$. 
This is easy to see as for $x, x' \in X$ such that $\pi(x) > \pi(x')$, if $(i_1, i, i_3, x')$ is a $(2413)$ copy, then $(i_1, i, i_3, x)$ is a $(2413)$ copy as well. 
As a reminder, we consider $4$-heavy copies, and hence ``1'' and ``4'' are in the same $j$-bucket and, therefore, ``1'' and ``3'' are in different $j$-buckets.

\paragraph{Approximating $|A_i|$.}
By definition, we have that $|A_i| = \sum_{x \in X} f(x)$.
Moreover, since $f(x)$ is monotone with respect to $\pi(x)$ over $x \in X$, we approximate $\sum_{x \in X} f(x)$ by applying \Cref{lem:Birge}. 
Let $X' \subseteq X$ be the subset of size $O(\log(n) / \eps)$ of indices, and corresponding to \Cref{lem:Birge}, for which is needed to (approximately) compute $f(x)$ for $x \in X'$.
Observe that all the elements in $X$ belong to a well-defined rectangle. Hence, each point in $X'$ can be found in $\poly \log n$ time.

\paragraph{Approximating $f(x)$.}
Let $x \in X'$. Similarly to before, we approximate $f(x)$ in $\poly(\log n, 1/\eps)$ time using the segment tree and another application of the Birg\'e technique. 

Indeed, let $S_x$ be the set of all $1$-candidates, which are elements between $i$ and the rightmost end of its $j$-bucket, and whose values are less than $\pi(x)$. 
For each $y \in S$, let $g(y)$ denote the number of $(2413)$-copies of the form $(i_1, i, y, x)$. 
Note that $g(y)$ is monotone non-increasing in value. That is, when $\pi(y)$ increases, the number of $(2413)$-copies in $A_i$ that $y$ participates in as a ``1'' can only decrease. 
Moreover, it is easy to compute $g(y)$ exactly for a specific value of $y$ by invoking a single operation specified by \Cref{lemma:segment-tree}. That operation would count all elements that are larger than $\pi(y)$, smaller than $\pi(x)$, and are located to the left of $i$.

Now, because of the monotonicity of $g$, and because $f(x) = \sum_{y \in S_x} g(y)$, we apply the Birg\'e technique (\Cref{lem:Birge}) to approximate $f(x)$ for any specific value of $x \in X$ to within a $(1+\eps/3)$-factor using $O(\log (n) / \eps)$ computations of $g(y)$. 

Each of the applications of Birg\'e introduces a multiplicative error of $1 \pm \eps/3$. Provided that $\eps \leq 1/2$, the total multiplicative error is less than $1 \pm \eps$. 
\end{proof}

\subsection{Approximating the Number of $3$-Heavy Copies}
\label{sec:approximating-3-heavy}

\begin{proof}[Proof of \Cref{lem:3-heavy-approx}]
The proof is similar to the one for \Cref{lem:4-heavy-approx}, except that in counting $2413$ copies, the algorithm fixes the ``3'' but not the ``4''. 
In this proof, we only consider $3$-heavy type-$j$ copies, where the ``$1$'' is located in the same $j$-bucket as the ``$3$'', while the ``$4$'' is located in a neighboring (to the left) $j$-bucket. 

Fix a ``3'' at location $i$. We consider candidates for ``$4$'' in the neighboring $j$-bucket -- these are all locations $x$ with a value higher than that in location $i$. 
Among these locations $x$, the count of $2413$ copies with the ``4'' at $x$ and the ``3'' at $i$, which are $3$-heavy, is ``monotone by location'': it becomes bigger as the index/location $x$ grows in the relevant bucket.
This is the case since the possibilities for a ``1'' remain fixed while moving the ``4'' to the right allows more options for a ``2''. 
So we apply Birg\'e for the first time here and only need to compute the approximate count for $O(\log(n) / \eps)$ specific $x$ locations.

Now, fixing a specific $x$, i.e., specific location of the ``$4$'', we proceed exactly as in the proof of \Cref{lem:4-heavy-approx}; this requires another application of Birg\'e. 
This completes the analysis.
\end{proof}
\section{Approximate 2D Segment Tree Data Structure}
\label{sec:counting-12-copies}

In this section, we create a near-optimal data structure that returns the number of $12$ copies inside an axis-parallel arbitrary rectangle, proving~\Cref{lemma:counting-12-copies} (restated for convenience).

\begin{proposition}[Restatment of~\Cref{lemma:counting-12-copies}]
    There exists a deterministic data structure for $n$-point sets in $\R^2$, that implements the following with preprocessing time $\tO(n \eps^{-1})$. Given an axis-parallel rectangle $R \subseteq \R^2$ as a query, the data structure reports a $(1+\eps)$-approximation of the number of $12$-copies inside $R$, with per-query time $\poly(\log n) \cdot \eps^{-1}$.
\end{proposition}

Note that the existing literature on approximate counting of $12$-copies \cite{AnderssonPetersson98, ChanPatrascu2010} does not build a data structure as described in \Cref{lemma:counting-12-copies}. While existing work approximates the total count of $12$-copies in only a single box -- the whole dataset of $n$ points -- our data structure preprocesses the $n$ points once and then allows us to retrieve the (approximate) number of $12$-pairs for arbitrary sub-rectangles we query, in polylogarithmic time.
In what follows, we describe how to build this data structure.

\subsection{Two-Dimensional Segment Tree}
\label{sec:2d-segment-tree-for-12-copies}
As a starting point, we build a two-dimensional \emph{sparse} segment tree over the points $(i, \pi_i)$ over all $i \in [n]$. 
In \Cref{sec:2D-segment-tree}, we recall the definition of a segment tree and describe how we use it to count $4$-patterns.
For counting $5$-patterns, we build a two-dimensional segment tree as follows:
\begin{enumerate}[(1)]
    \item A segment tree $S$ is built over the points $(i, \pi_i)$ with respect to their $x$-coordinate. We also use \emph{outer} segment tree to refer to $S$.
    \item Consider a vertex $v$ in $S$, and let $[a, b]$ be the interval of the $x$-axis $v$ corresponds to. Then, $v$ stores all the points $(i, \pi_i)$ such that $a \le i \le b$.
    \item The points inside each vertex $v$ of $S$ are organized as a segment tree with respect to the $y$-coordinate of the $v$'s points. We call these segment trees \emph{inner}.
    \item Let $v$ be a vertex in the outer and $w$ a vertex in the $v$'s inner segment tree. 
    Let $v$ correspond to $[a, b]$ and $w$ to $[c, d]$. 
    Then, $w$ stores all the points within rectangle $[a, b] \times [c, d]$, i.e., $w$ stores all $(i, \pi_i)$ such that $a \le i \le b$ and $c \le \pi_i \le d$.
    Two copies of those points are kept. One copy is sorted with respect to the $x$-coordinates and the other copy is sorted with respect to the $y$-coordinate.

    \item To ensure that $S$ fits into $n \cdot \poly(\log n)$ space, those vertices in an inner segment tree that do not contain any input point are \textbf{not created}. For example, in terms of implementation, they point to \textsc{null}.
\end{enumerate}
Hence, $S$ is a segment tree of segment trees. The outer segment tree partitions the plane with respect to the $x$-coordinate into recursively nested strips. The inner segment trees partition each of the strips into another family of recursively nested strips but with respect to the $y$-coordinate.

A point $(i, \pi_i)$ is replicated within $O(\log n)$ vertices in the outer segment tree. 
Each of those outer vertices replicates $(i, \pi_i)$ $O(\log n)$ times within its inner segment tree. 
Hence, a point is replicated $O(\log^2 n)$ times within $S$.

This conclusion has two implications. 
First, the points inside the vertices of the inner segment trees can be sorted in $O(n \log^3 n)$ time.
Second, the total number of non-empty vertices across all inner segment trees is $O(n \log^2 n)$. 
This is essential as it enables us to build and maintain $S$ in only $\tO(n)$ time by not creating the vertices that contain no point inside.

\subsection{Pre-processing 12-copy counts}
Once the two-dimensional segment tree $S$ is built as described, we process its vertices to pre-compute the number of $12$ copies inside each vertex of the inner segment trees. 
First, we show the following claim.

\begin{lemma}[12 copies across disjoint useful rectangles]
    \label{lem:disjoint-rectangle-12}
    Given two distinct vertices $w_1$ and $w_2$ belonging to inner segment trees of $S$, we can $(1+\eps)$ approximate the number of $12$ copies $i_1, i_2$ with $(i_1, \pi_{i_1})$ inside $w_1$ and $(i_2, \pi_{i_2})$ inside $w_2$ in $O(\eps^{-1} \log^2 n)$ time.
\end{lemma}
\begin{proof}
Let $R_1$ be the rectangle corresponding to $w_1$, and $R_2$ be the rectangle corresponding to $w_2$.
If $R_2$ is below or left of $R_1$, we just return $0$. 
If $R_2$ is up and right of $R_1$, we return $|R_1| \cdot |R_2|$.

Without loss of generality, assume that $R_2$ is to the right of $R_1$. We now use Birg\'e theorem to approximate the number of relevant $12$ pairs as follows.
Note that we simply need to return pairs of points $(u, v)$ from $R_1 \times R_2$ such that $u$ is below $v$. The higher up the $u$ is, the fewer $(u, v)$ pairs there are. Thus, we only need to compute the number of possible $v$, for $O(\eps^{-1} \log n)$ different possibilities of $u$ by \Cref{lem:Birge}, which can be effectively computed in $O(\log n)$ time each.
\end{proof}

\review{
The next simple lemma concerns $12$-counts inside vertices of the 2D segment tree. The proof follows from a simple double counting argument.}

\begin{lemma}\label{lem:precomputing-12}
    There is an algorithm that in $O(n \cdot \log^3 n)$ time exactly  computes the number of $12$ copies within \emph{each} vertex of the two-dimensional segment tree.
\end{lemma}

\review{We note that the polylog exponent can be improved to $2$ if one only wants approximate answers, or $2.5$ if one uses the state of the art (and more complex) algorithm of Chan and P\u{a}tra\c{s}cu \cite{ChanPatrascu2010} for $12$-counting. However, since the construction of the data structure requires time $O(n \log^3 n)$ anyway, these improvements would only complicate or weaken the result without asymptotically improving the running time.}

\begin{proof}
\review{As discussed above, each point is replicated $O(\log^2 n)$ times in the 2D segment tree. In other words, the total length of all list associated with vertices of the segment tree is $O(n \log^2 n)$. For each such vertex with list of length $t$, one can exactly count the number of $12$-copies in the vertex in time $O(t \log t)$, using the textbook merge sort based exact algorithm for $12$-counting (see, e.g., \cite{ChanPatrascu2010}). The total running time is $O(n \log^3 n)$.}
\end{proof}

\subsection{Proof of \Cref{lemma:counting-12-copies}}
We first build a two-dimensional segment tree $S$, as discussed above.
Second, following \Cref{lem:precomputing-12}, we precompute the $(1+\eps)$-approximate number of $12$ copies within each vertex of $S$. 

For each axis-parallel rectangle $R$ query, decompose $R$ into $O(\log^2 n)$ axis-parallel rectangles such that those rectangles correspond to the vertices of the segment tree $S$.
In \Cref{sec:2D-segment-tree}, we already discussed one such decomposition. 
Let $\cR$ be the set of rectangles obtained in that decomposition.
Each rectangle in $\cR$ can be located within $S$ in $O(\log n)$ time.

For each of the $O(\log^4 n)$ pairs of $(R_1, R_2) \in \cR \times \cR$, by \Cref{lem:disjoint-rectangle-12}, we approximate the number of $12$ copies across $R_1$ and $R_2$ in $O(\eps^{-1} \log^2 n)$ time. 
To the sum of those approximated counts, we also add the number of $12$ copies inside each rectangle in $\cR$.
Therefore, the number of $12$ copies inside $R$ can be approximated in $O(\eps^{-1} \log^6 n)$ time.
\section{All $5$-Patterns via Global Separators and $12$-Copies}
\label{sec:5-patterns}

The main goal of this section is to prove our result on approximately counting $5$-length patterns in near-linear time.
\maintheorem*

This section is split into three parts.
In \cref{sec:2413}, we have already described the idea of using a separator to induce additional structure among the copies of a fixed pattern. 
A part of this section details an extension of those ideas.
As a reminder, for counting $4$-heavy copies of the pattern $2413$, the idea was to first fix a candidate for the ``4'' and then to find a convenient way to separate the candidates for the ``1'' and the candidates for the ``3''.

First, in \cref{sec:global-separators}, we elaborate that it is also possible to \emph{first fix a separator} and then to fix a candidate for one of the positions in the pattern. In fact, we show that two separators can be fixed before fixing any candidate.
We refer to these separators by \emph{global}.

Second, in \cref{sec:proof-of-main-theorem}, we show how to leverage that type of separators to count copies of almost all $5$-patterns.
Our proof is computer-assisted. 
That is, we implemented a simple algorithm to test for which $5$-length patterns their counts can be approximate using the Birg\'e technique, global separators, and our data structure from \cref{sec:counting-12-copies}.
However, these ideas do not suffice to approximately count all possible configurations of the copies of $13524$, $14253$, and their symmetric patterns.
Nevertheless, also in \cref{sec:proof-of-main-theorem}, we describe how to reduce the count of ``difficult configurations'' of those patterns to counting the copies of $2413$ and $3142$.
That yields the proof of \cref{thm:main_theorem}.

Third, in \cref{sec:thm-enumeration}, we discuss our proof of \cref{thm:enumeration}.

\subsection{Global Separators}
\label{sec:global-separators}

We use the segment tree $S$ discussed in \Cref{sec:2d-segment-tree-for-12-copies} to describe our definition of global separators.
As a running example, consider the $5$-pattern $24135$.
For each vertex $v$ in the outer segment tree of $S$, we want to count all the copies of $24135$ that \textbf{do not} appear in $v$'s children, but do appear in $v$.
This naturally gives rise to the idea of \emph{vertical separators}.
Specifically, let $[a, b]$ be the range $v$ corresponds to. 
Saying that a copy belongs to $v$ \textbf{but not} to any of its children is equivalent to saying that there is a vertical separator at the $x$-coordinate $(a+b)/2$ such that the copy is on ``both sides'' of the separator at $(a+b)/2$. 
One such separator is depicted in \Cref{fig:global-separator}. 
Let $\cC_v$ be the set of all such copies.

We also partition the copies in $\cC_v$ with respect to their $y$-coordinates in a way similar to the horizontal separator. Namely, $\cC_v$ is partitioned into $\{\cC_{v, w}\ |\ \text{$w$ is a vertex in the $v$'s inner segment tree}\}$ such that if a copy belongs to $\cC_{v, w}$, then $w$ is the smallest vertex in the $v$'s inner segment tree that the copy belongs to.
This naturally induces \emph{horizontal separators} within $[a, b] \times [0, \infty)$.
After fixing a vertical separator and then a horizontal separator, we define a rectangle inside which we aim to count the copies of a fixed pattern. 
An example of such a rectangle is illustrated in \Cref{fig:global-separator}.

\paragraph{The number of vertical-horizontal separator pairs.}
From our construction, the fixing of a vertical and then a horizontal separator corresponds to a vertex in an inner segment tree of $S$. 
Since, as explained in \Cref{sec:2d-segment-tree-for-12-copies}, the total number of \textbf{non-empty} vertices across all inner segment trees is $O(n \log^2 n)$, there are $O(n \log^2 n)$ vertical-horizontal separator pairs that should be considered.
The segment tree $S$ itself defines those separator pairs, that is, a vertex in an inner segment tree defines one separator pair.

To execute this idea, we have to ensure that each copy is considered exactly once.

\paragraph{Counting each copy exactly once.}
Consider a vertex $v$ in the outer segment tree of $S$.
Now, let $w$ be a vertex in the inner segment tree of $v$.
As a running example, consider the $5$-pattern $24135$, and let $C$ represent a copy of this pattern.
To ensure that $C$ is contained in $v$ but not in its children, we require that the leftmost position, i.e., the position of ``2'', is to the left of the vertical separator, and that the position of ``5'' is to the right.

Similarly, to ensure that $w$ is the smallest vertex in $v$'s inner segment tree that contains $C$, we require that the topmost value, i.e., ``5'', is above, and the value of ``1'' is below the horizontal separator.
This setup is sketched in \Cref{fig:global-separator}.

To conclude, observe that $C$ is in the root of $S$, ensuring that $C$ is counted by some vertex $w$.
Second, for any copy $C$, there exists a unique vertex $w$ that counts $C$: only $w$ and its ancestors contain $C$, while none of $w$'s siblings do. 
This is because $w$'s siblings correspond to disjoint rectangles by the construction of segment trees.

\subsection{Proof of \cref{thm:main_theorem}}
\label{sec:proof-of-main-theorem}
The final piece of our proof is computer-assisted. We now provide additional details and discuss how to perform the computation efficiently.
We first describe an approach that enables us to approximately count almost all $5$-length patterns. 
After it, at the end of this proof, we discuss how to approximately count the remaining patterns.

\paragraph{Configurations.}
Our main algorithm counts $5$-patterns by distributing these counts across the inner vertices of $S$, as discussed in \cref{sec:global-separators}.
Once a vertex $w$ is fixed, we ensure that only copies that do not belong to any of the children of $w$ are counted, as explained next.

Fixing $w$ induces a horizontal and vertical separator.
For a fixed $w$, our algorithm considers all \emph{valid configurations}, such as: ``1'' and ``2'' are below the horizontal separator while ``3'', ``4'', and ``5'' are above; and ``2'' and ``4'' are to the left of the vertical separator, while ``1'', ``3'', and ``5'' are to the right.
We emphasize that for a configuration to be valid there has to be an element below and an element above the horizontal separator, and an element to the left and an element to the right of the vertical separator. This ensures that a counted copy does not belong to a child of $w$.

For each configuration, the algorithm fixes one element (e.g., the ``4'').
The choice of which element to fix is guided by our algorithm, which also provides a ``recipe'' on how to leverage the Birg\'e technique (\cref{lem:Birge}) and the $12$-copy primitive (\cref{sec:counting-12-copies}).
Each configuration naturally provides candidate locations for the elements of a pattern. This is discussed in more details in the next paragraph.

\begin{figure}[h]
    \centering
    \includegraphics[width=0.7\linewidth]{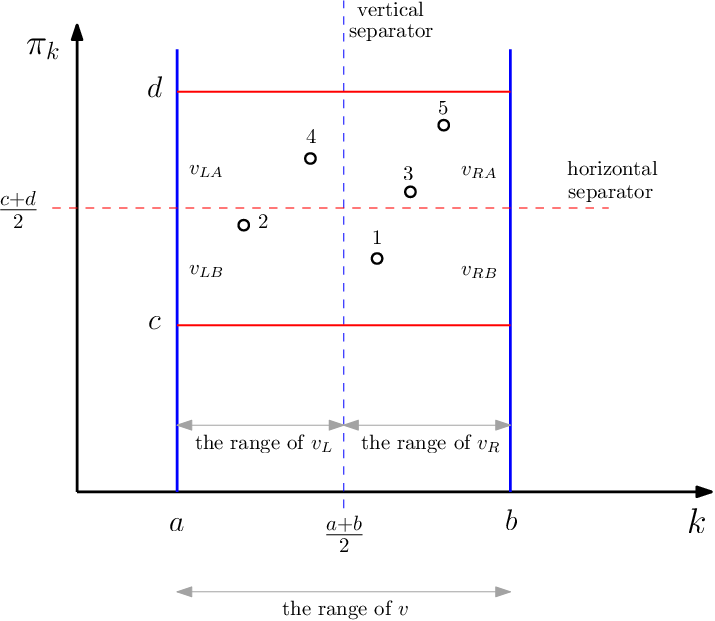}
    \caption{
        This is a more detailed example compared to the one provided in \cref{fig:global-separator} to aid the discussion in \cref{sec:proof-of-main-theorem}.
        Here, $v$ corresponds to the rectangle $[a, b] \times [0, \infty)$. Its two children, $v_L$ and $v_R$, correspond to rectangles $[a, (a + b)/2] \times [0, \infty)$ and $[(a + b)/2, b] \times [0, \infty)$.
    }
    \label{fig:global-separator-more-details}
\end{figure}
\paragraph{Candidate locations for an element.}
We now discuss where are candidate locations for ``1'', ``2'', ``3'', ``4'', and ``5'' with respect to a given configuration.
Let $v$ be the vertex in the outer segment tree which contains $w$. Let $v_L$ and $v_R$ be the two $v$'s children in the outer segment tree.
Let $[a, b] \times [c, d]$ correspond to $w$. Finally, define $v_{LB}$ to be the vertex in $v_L$'s segment tree corresponding to $\bb{a, \frac{a+b}{2}} \times \bb{c, \frac{c+d}{2}}$ and $v_{LA}$ corresponding to $\bb{a, \frac{a+b}{2}} \times \bb{\frac{c+d}{2}, d}$. 
Similarly, define $v_{RB}$ and $v_{RA}$ to correspond to $\bb{\frac{a+b}{2}, b} \times \bb{c, \frac{c+d}{2}}$ and $\bb{\frac{a+b}{2}, b} \times \bb{\frac{c+d}{2}, d}$, respectively.
It may be helpful to interpret ‘L’ as left, ‘R’ as right, ‘B’ as below, and ‘A’ as above the corresponding separators.
One such example is illustrated in \cref{fig:global-separator-more-details}.

Say that the algorithm fixes the ``4'' after choosing a configuration. Then, in the example in \cref{fig:global-separator-more-details}, it means that the algorithm iterates over points in $v_{LA}$ to select a candidate for the ``4''.
Similarly, when the Birg\'e technique is applied to consider the candidates for ``2'', it is applied within the points of $v_{LB}$, and so on.
Again, the choice of which element to fix in this configuration is made by the algorithm to enable the use of the Birg\'e technique (\cref{lem:Birge}) and the $12$-copy primitive (\cref{sec:counting-12-copies}).

\paragraph{Our recipe and how to read it.}
We give our recipe for counting $5$-length patterns at \cite{ourRecipeLink}.
We now illustrate how to read this recipe. For instance, consider the row 
\begin{center}
    ``\verb!12|345, horizontal below3 ---> 3, 1, 2, 4, 5!'';
\end{center}\review{it is sketched in \cref{fig:5-pattern-recipe-sample1}}. The bar (`$|$') represent the vertical separator. This configuration has the vertical separator between ``2'' and ``3'' and the horizontal separator between the same pair of elements. The first element after ``\verb!--->!'' is the element that is fixed. Fixing ``3'' places ``4'' and ``5'' in a fully determined rectangle. 
Moreover, ``1'' and ``2'' are placed in a fully determined rectangle already after the choice of the separators. So, the number of ``12'' pairs and ``45'' pairs are approximated using the 12-copy data structure.

\begin{figure}[h]
    \centering
    \includegraphics[width=0.8\linewidth]{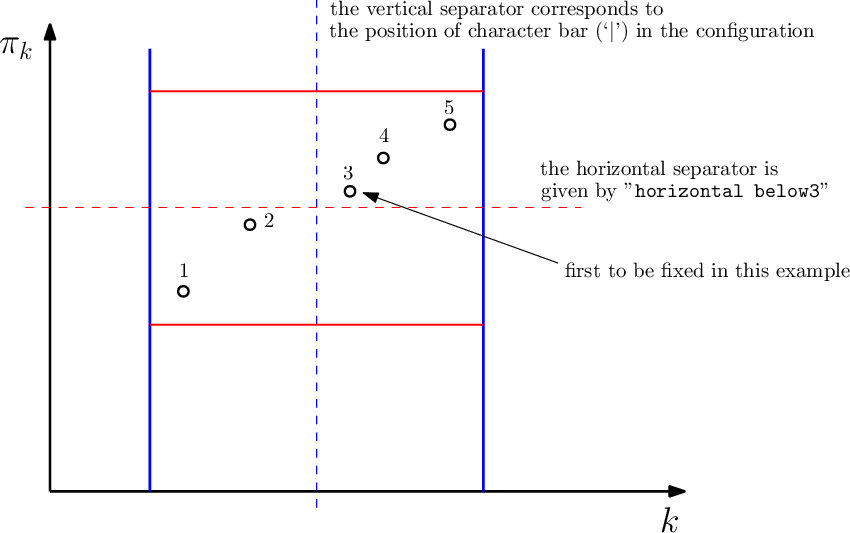}
    \caption{
        This figure corresponds to the row ``\texttt{12|345, horizontal below3 ---> 3, 1, 2, 4, 5}'' in our recipe at \cite{ourRecipeLink}.
        The \texttt{12|345} means we are considering permutation \texttt{12345}.
        The first element after ``\texttt{--->}'', i.e., ``3'' in this example, is fixed by trying all candidates for it.
    }
    \label{fig:5-pattern-recipe-sample1}
\end{figure}

Now, consider ``\verb!1|3425, horizontal below3 ---> 2, 1, 5, 3, 4!''\review{; it is sketched in \cref{fig:5-pattern-recipe-sample2}}. The vertical separator splits ``1'' and the remaining elements, while the horizontal separator is below ``3'' and above ``2''.
The ``2'' is fixed, i.e., the algorithm iterates over all the candidates for the ``2''. Over ``1'' it applies Birg\'e. Then, over ``5'' it applies Birg\'e. After that, ``3'' and ``4'' are in a fully determined rectangle and the count of ``34'' pairs is approximated using the 12-copy data structure.

\begin{figure}[h]
    \centering
    \includegraphics[width=0.8\linewidth]{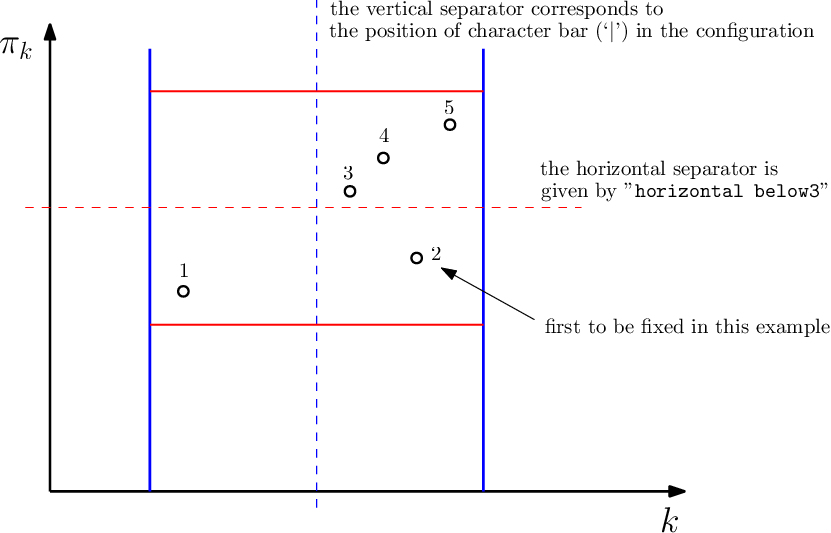}
    \caption{
        This figure corresponds to the row ``\texttt{1|3425, horizontal below3 ---> 2, 1, 5, 3, 4}'' in our recipe at \cite{ourRecipeLink}.
        The \texttt{1|3425} means we are considering permutation \texttt{13425}.
        The first element after ``\texttt{--->}'', i.e., ``2'' in this example, is fixed by trying all candidates for it.
    }
    \label{fig:5-pattern-recipe-sample2}
\end{figure}

\paragraph{Time complexity.}
Recall that our algorithm first fixes a vertical-horizontal separator pair. As discussed, there are $O(n \log^2 n)$ many of them to consider, each corresponding to a non-empty vertex of an inner segment tree of $S$.
Recall that the data structure described in \cref{sec:counting-12-copies}  approximately counts $12$ copies in a given axis-parallel rectangle in only $\poly(\eps^{-1}, \log n)$ time. 

Once a vertical-horizontal separator pair is fixed, for a given $5$-pattern, our algorithm considers all possible configurations. As discussed above, these configurations are of the form: ``1'' and ``2'' are below the horizontal separator; ``2'', ``3'', and ``5'' are to the right of the vertical separator.
For each configuration, the algorithm also generates a recipe on which entry to fix, i.e., for which among ``1'' through ``5'' to iterate over all the candidates, and how to utilize the Birg\'e technique and the $12$-copy primitive. 
Importantly for the running time, a recipe fixes only one element for a given configuration. 
Since $S$ contains $O(n \log^2 n)$ elements -- some elements might appear multiple times across different vertices of segment trees -- fixing one element of a configuration takes $O(n \log^2 n)$ time.
This running time of fixing an element is taken across all vertical-horizontal separator pairs corresponding to non-empty vertices of inner segment trees.

The rest of the counting is carried by applying the Birg\'e technique and using the $12$-copy primitive.
This leads to a total time complexity of $\tO(n \cdot \poly(\eps^{-1}))$.

\paragraph{When the recipe ``does not'' work: patterns $13524$ and $14253$.}
Above, we discussed one particular generic approach in approximating the number of copies of a given pattern. That approach can be summarized as follows. 
Given pattern $p$:
\begin{enumerate}[(1)]
    \item Fix a horizontal-vertical separator pair.
    \item Fix a configuration of $p$ with respect to the separator pair.
    \item Fix one element of $p$ and consider all its candidates within the corresponding rectangle, i.e., within the corresponding inner vertex of the segment tree $S$.
    \item Use the Birg\'e technique and $12$-copy data structure for the remaining elements.
\end{enumerate}
At \cite{ourRecipeLink}, we list the recipe on how to use the Birg\'e technique and $12$-copy counts data structure for each of the patterns, except for two equivalence classes.
Namely, there are precisely two equivalence classes for which the above recipe \emph{does not} work. 
These are the classes corresponding to the patterns $13524$ and $14253$ where, additionally, the vertical and horizontal separator appear right next to the ``1'' element. 
In other words, the ``1'' appears in the bottom-left area, and the rest of the pattern appears in the top-right area. One such example is shown in \cref{fig:5-pattern-failing-recipe}.
\begin{figure}[h]
    \centering
    \includegraphics[width=0.6\linewidth]{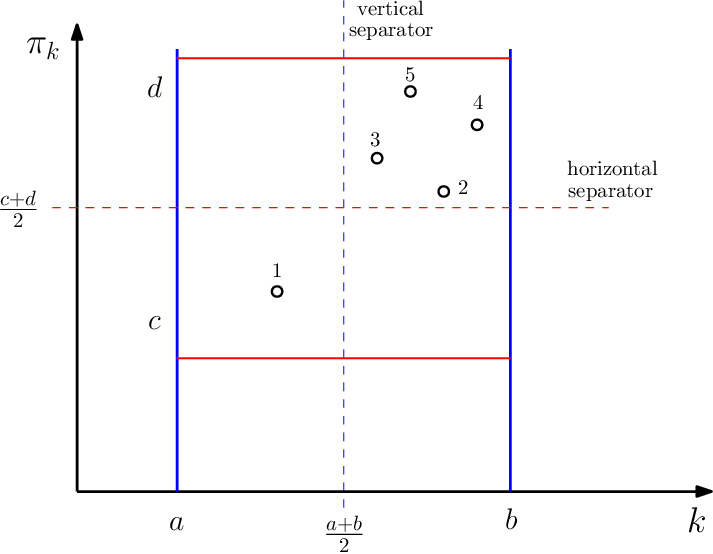}
    \caption{
        An example of a configuration of the $13524$ pattern that does not fall under the recipe discussed in \cref{sec:proof-of-main-theorem}.
        Nevertheless, the number of $3524$ tuples in the upper-right rectangle is exactly the number of $2413$ patterns in the same region.
        So, the number of $13524$ configurations as in this example equals the product of the number of elements in the bottom-left rectangle and the number of $2413$ copies in the upper-right rectangle. 
    }
    \label{fig:5-pattern-failing-recipe}
\end{figure}

We note that in these two cases, the top-right part of the pattern is order-equivalent to the pattern $2413$. 
Thus, it is possible to approximately count the number of $2413$-copies (or $3142$-copies, in the second case) in this top-right block in near-linear time, using our mechanism for approximate counting $4$-patterns. 
Counting the number of values in the bottom left in linear time is trivial.
Actually, the number of elements in the bottom-left area can even be fetched in $O(1)$ time, since the corresponding vertex in the segment tree $S$ already contains the list of all the elements in that rectangle; the size of that list is what the algorithm needs. 
The total count of $13524$ (or $14253$) copies in the full block is the product of these two quantities, and the proof follows.

\subsection{On the Proof of \cref{thm:enumeration}}
\label{sec:thm-enumeration}
To prove \cref{thm:enumeration}, we conduct the following modification to our approximate counting algorithm. Whenever the latter algorithm accesses (and/or aims to evaluate) elements from a monotone sequence $x_1 \geq x_2 \geq \ldots x_r$ using the Birg\'e technique, the enumeration algorithm will enumerate over all elements in the sequence one by one, starting from the largest value $x_1$ and the location corresponding to it in the input function, and descending in value through the sequence. It is straightforward to verify that, due to the monotonicity of all sequences of quantities considered, the algorithm will list all copies of the pattern throughout its run.

\review{We now discuss how to obtain the desired running time.
Observe that in our proofs, whenever we fix a prefix of elements in a pattern, the next element to be fixed or a 12-copy count to be performed is done within an axis-parallel rectangle.
Instead of fixing an element, e.g., ``3'', among its candidates at a position/value given by \cref{lem:Birge} by using data structure guaranteed by \cref{lemma:segment-tree}, we simply enumerate -- one by one -- over all the candidates for ``3'' in that data structure, and fix ``3''.
To ensure that the listing is done efficiently, we take three steps:
\begin{itemize}
    \item We perform enumeration starting with the candidates yielding the largest counts. 
    That way, we are not considering too many candidates that do not yield new copies to be listed.
    \item Let $s$ be an order in which our algorithm tells to fix the elements when counting, e.g., $s$ could mean: \emph{fix ``3'', then fix ``1'', then fix ``2''}. 
    When we fix candidates for a prefix of length $\ell$ of $s$, we first check whether there is a positive count of copies with those candidates fixed. 
    If there is not, we stop attempting to list copies with that prefix or the current prefix of length $\ell-1$ of $s$, but instead consider the next candidate for the $(\ell-1)$-st element that $s$ tells to be fixed. 
    That way, we ensure that each fixing of candidates that does \textbf{not} yield new copies can be charged to a sequence of candidates -- defined by a prefix of $s$ -- that has led to listing new copies.
    Observe that if our algorithms lists $t$ copies, then there are at most $5 \cdot t$ candidate-sequences to charge to. 
    Each of those sequences receives at most one charge. There is one exception to this: fixing a candidate corresponding to the prefix of $s$ of length $1$ that does not yield any new copies cannot be charged to a candidate-sequence of length $0$. 
    Nevertheless, there are at most $n$ different candidate-sequences of length $1$, and testing whether one of them leads zero counts or not takes $\poly(\log n)$ time.
    \item To obtain efficient listing of $12$-copies within an axis parallel rectangle, we perform a simple augmentation of our $12$-copy data structure.
    Consider a counting-query: ``How many $12$-copies a node $v$ of our data structure contains?''
    For a given node $v$, that count is precomputed and can be returned in $O(1)$ time.
    However, how to perform enumeration of those copies? 
    The challenge is that too many candidates for the ``1'' might not lead to any $12$-copy within $v$.
    So, how to -- during the listing stage -- detect the candidates for ``1'' that yield at least one $12$-copy within $v$? 
    To address that, we augment our data structure as follows: during the pre-processing stage, each candidate for ``1'' yielding some $12$-copies within $v$ is added to a list stored at $v$, so that the relevant candidates for ``1'' can all be easily enumerated during the listing stage.
    Once a candidate for the ``1'' is fixed within $v$, then all elements to the top-right within $v$ from that candidate yield $12$-copies.
    All such elements can be easily enumerated within the children of $v$, in $\poly(\log n)$ time per an enumerated element.
\end{itemize}
}

We note that Albert, Aldred, Atkinson, and Holton \cite{AAAH2001} employed a somewhat similar technique to construct a near-linear algorithm for the detection variant, specifically for the case $k=4$; see Section 4 in their paper.

\section*{Acknowledgments}
The authors wish to thank Sergio Cabello, Nishant Dhankhar, Talya Eden, Tomer Even, Chaim Even-Zohar, and Seth Pettie for fruitful discussions. We also thank anonymous reviewers for many useful suggestions, as well as for pointing out the works of \cite{HarPeled2006} and \cite{Halman2014} and their connection to the Birg\'e decomposition.

\bibliographystyle{alpha}
\bibliography{references}

\appendix

\section{Segment Trees}
\label{sec:2D-segment-tree}

We represent a permutation $\pi$ as a set of $n$ points $(i,\pi_i)$ for each $i \in [n]$. Building a sparse segment tree over these points to allow for two-dimensional counting queries is a textbook problem.
For completeness, we outline the construction and query support of this data structure as per \Cref{lemma:segment-tree}.

\paragraph{Building the segment tree.}
\newcommand{\tn}{\tilde{n}}
We aim to build a data structure to answer the two-dimensional queries \Cref{lemma:segment-tree} requires.
To achieve this, we construct two segment trees: one for the points $(i, \pi_i)$ and one for the points $(n - \pi_i, i)$, for all $i \in [n]$. We use $S$ to refer to the first one, while we use $\tS$ to refer to the second one. We now describe how to build $S$; the tree $\tS$ is built analogously.

Let $\tn = 2^{\lceil \log n \rceil}$.
We first build a segment tree $S$ on the $x$-coordinates, covering the range $[1, \tn]$. That tree can be visualized as a complete binary tree on $\tn$ leaves. 
The root vertex corresponds to the entire interval, its left child to the interval $[1, \tn / 2]$, and its right child to the interval $[\tn/2 + 1, \tn]$. 
In general, if a vertex corresponds to the interval $[t, t + 2^j - 1]$, its left and right children correspond to the intervals $[t, t + 2^{j - 1} - 1]$ and $[t + 2^{j - 1}, t + 2^j - 1]$, respectively.

Second, consider a vertex $v$ in $S$ and let $[a, a + 2^j - 1]$ be the range $v$ corresponds to.
The vertex $v$ stores in an array $A_v$ all the points $(i, \pi_i)$ such that $a \le i \le a + 2^j - 1$.
$A_v$ is sorted with respect to the $y$-coordinates, i.e., with respect to $\pi_i$.

It is folklore, and also easy to prove, that a point $(i, \pi_i)$ is stored in $\log \tn$ vertices $v$ of $S$. Therefore, a point $(i, \pi_i)$ is replicated $\log \tn$ times inside $S$.

To populate $S$ from $\pi$, we insert the points $(i, \pi_i)$ one by one, adding them to a list $L_v$ for each vertex $v$ covering the corresponding range.
After the insertions, each $L_v$ is then sorted to form the array $A_v$.
There are $O(n \log n)$ points in $S$, partitioned across different $A_v$. 
Hence, sorting all of them takes $O(n \log^2 n)$ time.

\paragraph{Implementing desired operations.}
\Cref{lemma:segment-tree} specifies three operations that need to be supported on $S$ and $\tS$.

The first operation counts the points within the rectangle $[i, j] \times [a, b]$. 
The range $[i, j]$ can be partitioned into $O(\log n)$ disjoint ranges, each associated with a vertex in $S$.
For each vertex $v$, we count points in $A_v$ with $y$-coordinates in $[a, b]$ using two binary searches, each in $O(\log n)$ time.
Hence, this operation can be implemented in $O(\log^2 n)$ time.

For the second operation, let $v_1, \ldots, v_k$ denote the $O(\log n)$ vertices in $S$ covering disjoint subranges that collectively form $[i, j]$.
Let $\ell$ be the index as described in the second operation, i.e., we are looking for the $\ell$-th leftmost element in $S^{a, b}_{i, j}$. 
To implement this, the algorithm finds the largest $k'$ such that $k' \le k$ and the cumulative number of points \textbf{within} $[i, j] \times [a, b]$ across $v_1, v_2, \ldots, v_{k'}$ is less than $\ell$, denoted by $\ell'$.
Next, we search for the $(\ell - \ell')$-th leftmost point within the left and right children of $v_{k'+1}$.
This approach processes $O(\log n)$ vertices in $S$, each performing two binary searches, for a total time complexity of $O(\log^2 n)$

The third operation on $S$ is equivalent to querying $\tS$ as in the second operation.

\end{document}